\documentclass[3p]{elsarticle} %3p,draft
 \myfooter[C]{}
\renewcommand{\thispagestyle}[1]{}
\usepackage{booktabs}
\makeatother

\date{}
\def\texpsfig#1#2#3{\vbox{\kern #3\hbox{\includegraphics{#1}\kern #2}}\typeout{(#1)}}
\usepackage{url,hyperref}
\hypersetup{
    colorlinks=true,
    linkcolor=blue,
    filecolor=magenta,      
    urlcolor=cyan,
    pdfpagemode=FullScreen,
    }
\usepackage{bbm}
\usepackage{eurosym}                           % Symbol for Euro
\usepackage[latin1]{inputenc}                  % Input encoding for Western Europe
\usepackage{url}                               % Display URLs in typewriter font
\usepackage{longtable}                         % Tables running of several pages
\usepackage{array}                             % more table options
\usepackage{graphicx,color}
\usepackage{amsthm}
\usepackage{amsbsy}
\usepackage{amssymb}
\usepackage{amsmath}
\usepackage{enumerate}
\usepackage{graphicx,color}
\usepackage{epsfig}
\usepackage{multirow,bigdelim}
\usepackage[table]{xcolor}% http://ctan.org/pkg/xcolor
\usepackage{natbib}
\usepackage{cleveref}
\usepackage{float}
\usepackage{algorithm}
\usepackage{algpseudocode}

\usepackage{mathtools}
\theoremstyle{plain}
\newtheorem{theorem}{Theorem}[section]
\newtheorem{cor}{Corollary}[section]
\newtheorem{dfn}[theorem]{Definition}

\newtheorem*{rem}{Remark}
\theoremstyle{remark}

\theoremstyle{plain}
\newtheorem{lem}[theorem]{Lemma}

  {\par\noindent\textit{Proof of the claim#1.}\quad}
  {\hfill$\square$\par}

\theoremstyle{definition}

\newcommand{\e}{{\rm e}}        % "e" number
\def\R{\mathbb{ R}}             % Real number
\def\E{\mathbb{ E}}             % Expectation
\def\Q{\mathbb{ Q}}  
\def\F{\mathcal{F}}             % Filtration
\def\var{\mathbb{V}\text{ar}}   % Variance
\renewcommand{\d}{{\rm d}}      % straight "d" in in integration and ODEs, \int_a^b f(x)\d x
\def\dW{{\rm d}W}               % dW in SDEs- Brownian noise.
\def\dt{{\rm d}t}

\DeclareMathOperator*{\arginf}{arg\,inf}
\def\ds{{\rm d}s}
\def\du{{\rm d}u}

\newcommand{\qv}[1]{\langle #1 \rangle}
\newcommand{\bs}[1]{\boldsymbol{#1}}
\def\1{{\mathbbm{1}}}            % Indicator function

\theoremstyle{plain}% default

\usepackage[margin=1cm]{caption}
\numberwithin{equation}{section}	     %Equation numbering per section
\title{\raggedright Lifted Heston Model: Efficient Monte Carlo Simulation with Large Time Steps}

\begin{document}
\author[1]{NICOLA F. ZAUGG \corref{cor1}}
\ead{N.F.Zaugg@uu.nl}
\author[1,2]{\raggedright LECH A. GRZELAK}
\ead{L.A.Grzelak@uu.nl}
\cortext[cor1]{Corresponding author.}
\address[1]{Mathematical Institute, Utrecht University, Utrecht, the Netherlands}
\address[2]{Financial Engineering, Rabobank, Utrecht, the Netherlands}
\begin{abstract}
The lifted Heston model~\cite{abi2019lifting} is a stochastic volatility model emerging as a Markovian lift of the rough Heston model~\cite{el2019characteristic} and the class of rough volatility processes. The model encodes the path dependency of volatility on a set of $N$ square-root state processes driven by a common stochastic factor. While the system is Markovian, simulation schemes such as the Euler scheme exist, but require a small-step, multidimensional simulation of the state processes and are therefore numerically challenging. We propose a novel simulation scheme of the class of implicit integrated variance schemes~\cite{jaber2024simulation,jaber2025simulating}. The method exploits the near-linear nature between the stochastic driver and the conditional integrated variance process, which allows for a consistent and efficient sampling of the integrated variance process using an inverse Gaussian distribution. Since we establish the linear relation using a linear projection in the $L^2$ space, the method is optimal in an $L^2$ sense and offers a significant efficiency gain over similar methods. We demonstrate that our scheme achieves near-exact accuracy even for coarse discretizations and allows for efficient pricing of volatility options with large time steps.
\noindent 
\end{abstract}
\maketitle
{\let\thefootnote\relax\footnotetext{The views expressed in this paper are the author's personal views and do not necessarily reflect the views or policies of their current or past employers. The authors have no competing interests.}}
{\let\thefootnote\relax\footnotetext{An accompanying Python code to the article containing examples is available on \href{https://github.com/NFZaugg/C-LP-Lifted-Heston}{Github}.}}
% {\let\thefootnote\relax\footnotetext{An accompanying Python code to the article containing examples is available on \href{https://github.com/NFZaugg/ImpliedVolatilityRandomization}{Github}.}}
%%%%%%%%%%%%%%%%%%%%%%%%%%%%%%%%%%%%%%%%%%%%%%%%%%%%%%%%%%%%
\section{Introduction}

Stochastic volatility models are a core class of models in mathematical finance to price exotic derivatives~\cite{Oosterlee_Grzelak_2020, bergomi2015stochastic}. Their flexibility enables robust calibration to most market regimes, and the models are able to price exotic derivatives consistently. The rise of rough volatility models, i.e., models whose volatility process resembles a fractional Brownian motion, promised to further advance the capabilities of the models by incorporating the inherent mean-reverting nature of volatility~\cite{bayer2016pricing, bayer2023rough}. These models can capture certain features of volatility which regular stochastic volatility models cannot capture, for instance, the sharp term structure of the at-the-money implied volatility~\cite{el2019characteristic}. Such rough volatility models emerged in various shapes, all based on the idea of a fractional driver for volatility.  A prominent rough model is the so-called \emph{``rough Heston model"}~\cite{el2019characteristic}, based on the classical Heston model~\cite{heston1993closed} with a CIR-type mean-reversion process for volatility. The drift and diffusion terms are then adjusted using a fractional kernel to add a rough structure to the paths of the volatility. The volatility process $\tilde{V}_t$ is defined as the equation
\begin{equation}
\label{eq:rough_heston}
    \tilde{V}_t = \tilde{V}_0 + \int_{t_0}^t \frac{(t-s)^{H-0.5}}{\Gamma(H+0.5)} \lambda(\theta-\tilde{V}_s) \,\ds + \nu\int_{t_0}^t \frac{(t-s)^{H-0.5}}{\Gamma(H+0.5)}  \sqrt{\tilde{V}_s}\,\dW_s,
\end{equation}
where $\Gamma$ is the gamma function, $\lambda,\theta, \nu, V_0$ are the usual Heston parameters, and $H \in (0,\frac{1}{2}]$ the so-called \emph{``Hurst"} coefficient, determining the roughness of the volatility paths.
While theoretically compelling, rough models face significant computational challenges due to the non-Markovian structure of rough volatility models and the lack of established numerical frameworks to address those challenges. The standard method is therefore to approximate the rough models with a Markovian system of differential equations. The fractional kernel $\frac{(t-s)^{H-0.5}}{\Gamma(H+0.5)}$, which causes the non-Markovian property, can be expressed as an infinite sum of exponentials and then truncated to finite terms~\cite{bayer2023markovian}. This approximation naturally leads to a Markovian multidimensional factor model, where each factor is an auxiliary process (or state processes) representing one summation term of the truncated exponential sum. For the rough Heston model, the multidimensional model was then established as the \emph{lifted Heston model}~\cite{abi2019lifting}, representing a \emph{Markovian lift} of the non-Markovian rough Heston model. 
\subsection{Lifted Heston Model}
The lifted Heston model as described in~\cite{abi2019lifting} is a generalization of the classical Heston model, and emerges as a natural Markovian approximation of the rough Heston model. It describes the instantaneous volatility of the asset process as a sum of $N$ mean-reverting processes, which share a common diffusion term.  Let $S_t,V_t$ be stochastic processes on a probability space $(\Omega, \mathbb{F} := (\mathcal{F}_t)_{t \geq t_0}, \Q)$, where $S_t$ denotes the price of a financial asset at time $t \geq t_0$, and $V_t$ is its variance process. Furthermore, let $r$ be a fixed interest rate. The model with $N$ states is defined as the following set of stochastic differential equations:
\begin{align}
    \label{eq:lifted_heston}
    \frac{\d S_t}{S_t} &= r\d t + \sqrt{V_t} \d W^1_t,  \\
    \label{eq:lifted_heston_middle}
    V_t &= g_0(t) + \sum_{n=1}^N \omega_n U^n_t =: g_0(t) + \boldsymbol{\omega} \cdot \bs U_t,\\
    \label{eq:lifted_heston_end}
    \d U^n_t &= (-x_nU_t^n - \lambda V_t)\d t + \nu \sqrt{V_t}\d W_t^2, 
\end{align}
where $W_t^1,W_t^2$\footnote{Since most calculations in the paper concern $W_t^2$, we will mostly only write $W_t$.} are correlated Brownian motions with $\d W_t^1 \d W_t^2 = \rho \d t$ with $\rho \in [-1,1]$, and $\boldsymbol{\omega }= (\omega_1, \omega_2 \dots, \omega_N) \in \R^N_{\geq 0}, \boldsymbol{x} = (x_1, x_2 \dots, x_N) \in \R^N_{\geq 0}$ are fixed vectors of size $N$. The further parameters $\nu ,\lambda \in \R_{> 0}$ are the vol-of-vol and the mean-reversion speed. The state processes $U_{t}^n$ are displayed in matrix notation as the column vector $\bs{U}_t = (U_{t}^1,U_{t}^2,\dots,U_{t}^N)^T$ and have initial value $\bs{U}_{t_0} = \bs{0}_N$ at $t_0$. These processes act as a state retainer for the past path of the volatility (hence we refer to them as ``state" processes). The function $g_0(\cdot)$ is a deterministic map that describes the shape of $\E[V_t], t \geq t_0$, and can be fit to the market. The method presented in this article does not prescribe a particular choice of $g_0(\cdot)$. We will use the function~\cite{abi2019lifting}
\begin{equation}
\label{eq:fwd_curve_simple}
    g_0(t) = V_0 + \lambda \theta \sum_{n=1}^N \frac{\omega_n}{x_n} \left[1 - \e^{-x_n (t-t_0)}\right] ,
\end{equation}
where the additional parameter $\theta \geq 0$ is used to steer the long-term variance mean, and $V_0 \geq 0$ is the deterministic variance at $t = t_0$. With this initial variance curve, the expected value of $V_T$ at a time $T \geq t_0$ is given by
\begin{equation}
\label{eq:expecation}
    \E[V_T] = V_0 + \sum_{n=1}^N \omega_n \int_{t_0}^T \lambda(\theta - \E[V_t]) \e^{-x_n (T-t)}\dt.
\end{equation}
We refer to \ref{app:proof_mean} for the exact calculation of this expression. The lifted Heston model is a generalization of the regular Heston model, and collapses with the choice $N=1, \omega_1 = 1,x_1 = 0$ and $g_0(t) = V_0 + \lambda\theta (t-t_0)$. In this case, we obtain the system
\begin{align}
    \label{eq:heston}
    \frac{\d S_t}{S_t} &= r\d t + \sqrt{V_t} \d W^1_t,  \\
    \label{eq:heston_end}
    V_t &= g_0(t) +  U^1_t,\\
    \d U^1_t &= - \lambda V_t \d t + \nu \sqrt{V_t}\d W_t^2, 
\end{align}
from which we can derive the dynamics of $V_t$ as 
\begin{equation}
\label{eq:heston_simple}
    \d V_t = \lambda (\theta - V_t)\dt + \nu \sqrt{V_t}\d W_t^2, \quad V_{t_0} = V_0.
\end{equation}
In the formulation of \cref{eq:lifted_heston,eq:lifted_heston_middle,eq:lifted_heston_end} the lifted Heston model with $g_0$ as (\ref{eq:fwd_curve_simple}) has
$(5 + 2N)$ parameters (for a fixed $N$), since for each $n \leq N$, values for $x_n$ and $\omega_n$ need to be set. To reduce the number of parameters for a large number of state processes, we introduce a parametrization~\cite{abi2019lifting} of the parameters $\bs \omega$ and $\bs x$ to reduce the number of parameters to a maximum of $6$ for any number of state processes. The core idea is that the parameters for $\bs{\omega}$ and $\bs x$ can be expressed as an additional parameter $H$, which resembles the Hurst index determining the ``roughness" of the paths of the rough Heston variance process of \Cref{eq:rough_heston}. Let $H \in (0, \frac{1}{2})$ be the Hurst index. Then, we define 
\begin{equation}
    r_N = 1+ 10N^{-0.9}, \quad N \geq 1.
\end{equation}
From this value of $r_N$ we define the mean-reversion weights $\bs{\omega}$ and $\bs x$ for a fixed $N$ as
\begin{equation}
    \omega_n = \frac{(r_N^{\frac{1}{2}-H}-1)r_N^{(H-\frac{1}{2})(1+\frac{1}{2}N)}}{\Gamma(H+\frac{1}{2})\Gamma(\frac{3}{2}-H)} r_N^{(\frac{1}{2}-H)n}, \quad x_n = \frac{\frac{1}{2}-H}{\frac{3}{2}-H}\frac{(r_N^{\frac{3}{2}-H}-1)}{(r_N^{\frac{1}{2}-H}-1)}r_N^{n-1-\frac{N}{2}}\quad n \leq N.
\end{equation}
As shown in~\cite[Theorem A.2]{abi2019lifting}, this parametrization ensures that the lifted Heston model converges to a rough Heston model with Hurst index $H$, as $N \to \infty$, meaning that $V_t \to \tilde{V}_t$.
\subsection{Article Aim}
 The numerical challenges inherent in rough volatility models motivate the development of new and innovative algorithms, both for rough volatility itself and for Markovian counterparts such as the lifted Heston model~\cite{ma2022fast,bayer2023markovian,jaber2025simulating,richard2023discrete}. In this paper, we introduce a novel efficient sampling scheme for the lifted Heston model based on the integrated variance implicit (IVI) schemes~\cite{jaber2024simulation,jaber2025simulating}. The novel method obtains the optimal implicit equation between integrated variance and diffusion process of the volatility by utilizing a linear projection method. It is capable of simulating much larger time steps than a classical Euler method, provides robust (positive) samples for the variance process at all times, and outperforms other IVI schemes significantly in terms of convergence. A Python implementation of the algorithm can be found on \href{https://github.com/NFZaugg/C-LP-Lifted-Heston}{Github}. In this section, we introduced the stochastic setup of the lifted Heston model. \Cref{sec:scheme} covers a brief literature review for (lifted) Heston model schemes, and discusses the theoretical aspects of the novel scheme. It contains the main theorems to derive the required quantities for the method. In \Cref{sec:num}, we then present numerical examples to illustrate the effectiveness of the new method. 

\section{A Monte Carlo Scheme of Lifted Heston Model}
\label{sec:scheme}
The classical method to simulate a system of SDEs, such as the lifted Heston model, is the Euler-Maruyama (EM) scheme. The EM scheme iterates forward in time to obtain the paths and converges strongly in $L^2$ with an increasing number of steps, as $\Delta t$ goes to 0. Since the convergence depends on the size of the time steps, the scheme is not capable of efficiently simulating \emph{larger time steps}, where fewer simulation grid points are required to simulate the paths. As in the standard Heston model, larger time steps in the EM scheme can cause the variance process $V_t$ in (\ref{eq:lifted_heston_middle}) to become negative due to a large negative observation of $\d W^2_t$, leading to defective paths. This issue can be solved by introducing a volatility-path reflection or absorption at $0$, at the cost of introducing a bias in the simulation~\cite{deelstra1998convergence,alfonsi2005discretization,lord2010comparison}.

The connection of the lifted Heston model to square-root processes offers the possibility of applying the well-established research on one-dimensional square-root processes, such as the regular Heston model, to the lifted Heston model. The Heston model, whose volatility process is a CIR process, has various numerical approaches and insights available in the literature, which are relevant to the lifted Heston model. The exact simulation scheme of Broadie and Kaya~\cite{broadie2006exact} introduced a new approach, where the sampling is focused on the \emph{integrated variance} directly, rather than increments of the Gaussian noise process, as it would be usual for Euler-type schemes. Although this sampling proved to be computationally expensive, as it requires path-wise sampling from an expensive distribution, their research was improved using a cheaper stochastic collocation sampler~\cite{grzelak2019stochastic}, and inspired simpler, non-exact, yet highly efficient schemes based on similar principles~\cite{andersen2008simple, jaber2024simulation}. The derivations of Broadie and Kaya showed that large time step sampling is possible for stochastic volatility models when the integrated variance is sampled directly.
Abi Jaber and Attal~\cite{jaber2025simulating} then suggested an integrated variance framework in a general framework for rough volatility models, utilizing an inverse Gaussian sampling scheme for faster convergence compared to a Gaussian sampling scheme. The core idea is to establish an implicit relation between the integrated variance process $X_{s,t} := \int_s^t V_{u} \du$ and the diffusion term $Z_{s,t} := \int_s^t \sqrt{V_u} \d W^2_u$ for two fixed times $s<t$, from which one can derive an explicit distribution of $X_{s,t}$. This type of algorithm is called integrated variance implicit (IVI), and has been shown to be viable for the regular Heston model~\cite{jaber2024simulation}. However, as the method relies on an endpoint discretization of a time integral, the linear approximation that governs the implicit relation becomes increasingly inaccurate for larger time steps, bounding its viability as a large time step scheme.

In this section, we introduce a novel IVI scheme called the \emph{constrained linear projection (C-LP) scheme} for the lifted Heston model that is capable of taking larger time steps. The scheme derives the implicit relation through a linear projection onto a subspace in $L^2$, yielding an optimal linear approximation. The method avoids the trivial integral approximation, and we prove it to be the best linear approximation between $X_{s,t}$ and $Z_{s,t}$ in an $L^2$ sense. The linear projection is derived by calculating conditional moments of $X_{s,t}$ and $Z_{s,t}$, which we obtain directly from the stochastic processes. 
\subsection{Explicit solution with integrated variance}
\label{sec:explicit_sol}
The starting point for the linear projection scheme is the explicit representation of the log-price process $\log S_t$. We apply a log transformation to the stochastic differential equation of (\ref{eq:lifted_heston}), obtaining the SDE
\begin{equation}
    \d \log S_t = \big(r- \frac{1}{2}V_t\big)\dt + \sqrt{V_t}\dW^1_t.
\end{equation}
In its explicit form, this SDE is equivalent to
\begin{equation}
    \log \frac{S_t}{S_s} = \int_{s}^t \big(r-\frac{1}{2}V_u\big) \du + \int_{s}^t \sqrt{V_u}\dW^1_u, \quad t_0 <s<t.
\end{equation}
Since the Brownian motions $W_t^1,W_t^2$ are correlated with a coefficient $\rho$, we can find orthogonal Brownian motions $\tilde{W}_t^1, \tilde{W}_t^2$, such that $W_t^2 \stackrel{d}{=} \tilde{W}_t^2$ and $W_t^1 \stackrel{d}{=} \rho \tilde{W}_t^2 + \sqrt{(1-\rho^2)} \tilde{W}_t^1$, such that we can write the equation as
\begin{equation}
    \log \frac{S_t}{S_s} \stackrel{d}{=} (t-s)r-\frac{1}{2}\int_{s}^tV_u \du + \rho\int_{s}^t \sqrt{V_u}\d\tilde{W}^2_u + \sqrt{(1-\rho^2)}\int_{s}^t \sqrt{V_u}\d\tilde{W}^1_u.
\end{equation}
This shows that the log price $S_t$ given $S_s$ is given by the deterministic interest rate drift and three stochastic quantities. The second term $\int_{s}^tV_u \du $ is the integrated variance between $s$ and $t$. The third and fourth terms are stochastic integrals with quadratic variation 
\begin{equation}
    \left\langle\int_{s}^t \sqrt{V_u}\d\tilde{W}^2_u\right\rangle = \int_{s}^t V_u \d u,  \quad \left\langle\int_{s}^t \sqrt{V_u}\d\tilde{W}^1_u\right\rangle = \int_{s}^t V_u \d u.
\end{equation} Since the integrand $\sqrt{V_u}$ is independent of $\tilde{W}^1_u$, we can sample the third term from a normal distribution once we have computed the integrated variance $\int_{s}^t V_u \d u$. The forward simulation $S_t$ therefore reduces to obtaining a simulation of the integrated variance $X_{s,t} := \int_{s}^t V_u \d u$ and the stochastic integral $Z_{s,t} := \int_{s}^t \sqrt{V_u}\d\tilde{W}^2_u$. We can then write the log price process as
\begin{equation}
\label{eq:main_process}
    \log \frac{S_t}{S_s} \stackrel{d}{=} (t-s)r-\frac{1}{2}X_{s,t} + \rho Z_{s,t} + \sqrt{(1-\rho^2)X_{s,t} }\,N_{s,t}, \quad N_{s,t} \sim \mathcal{N}(0,1).
\end{equation}
 For a defined integrated state process $X^n_{s,t} := \int_{s}^t U^n_u \d u, n \leq N$, we express the variance increment in terms of these variables:
 \begin{equation}
 \label{eq:main_variance}
     V_t = V_s + g_0(t)- g_0(s)  + \sum_{n=1}^N -x_n X^n_{s,t} - \lambda X_{s,t} + \nu Z_{s,t}.
 \end{equation}
The simulation scheme therefore requires sampling of the quantities $X_{s,t}, Z_{s,t}$ and $X^n_{s,t}, n \leq N$. The core idea is to sample $X_{s,t}$ through an inverse Gaussian distribution with given parameters, and then derive the values for $Z_{s,t}$ and $X^n_{s,t}, n \leq N$ from these samples. We derive the distribution parameter for the inverse Gaussian distribution. The main observation is that the quadratic variation of the stochastic integral $Z_{s,t}$ is exactly equal to $X_{s,t}$. The objective of the scheme is then to obtain a \emph{linear equation} between the two random variables $X_{s,t}$ and $Z_{s,t}$, which enables a stopping time interpretation of the two quantities.  Suppose we assume two processes $t \mapsto \hat{X}_{s,t}, t\mapsto \hat{Z}_{s,t}$, such that the quadratic variation of $\hat{Z}_{s,t}$ is equal to $\hat{X}_{s,t}$. Furthermore, assume a linear transformation for the random variables $\hat{X}_{s,t}$ and $\hat{Z}_{s,t}$, i.e. coefficients $\alpha, \beta$, such that
\begin{equation}
    \label{eq:linear}
    \hat{X}_{s,t} = \alpha + \beta \hat{Z}_{s,t}, \quad a.s.
\end{equation}
In this case, since the quadratic variation $\qv{\hat{Z}}_{s,t}$ is almost surely equal to $\hat{X}_{s,t}$, the process $t \mapsto \hat{Z}_{s,t}, t\geq s$ is a local martingale, as it is a stochastic integral, and has a quadratic variation process $t \mapsto \hat{X}_{s,t}$.  We can then leverage the Dubins-Schwarz theorem~\cite{dubins1965continuous}, which states the following: 
\begin{lem}[Dubins-Schwarz~\cite{dubins1965continuous}]
Let $t\mapsto \hat{Z}_{s,t}$ be a local martingale with quadratic variation process $\qv{\hat{Z}}_{s,t} = \hat{X}_{s,t}$. Let $T_t, t \geq s$ be a sequence of random times $T_t = \inf\{u : \qv{\hat{Z}}_{s,u} > t\}$. Then, the process $$\hat{W}_t := \hat{Z}_{s,T_t},$$ is a Brownian motion and we have that $\hat{Z}_{s,t} = \hat{W}_{\qv{\hat{Z}_{s,t}}} = \hat{W}_{\hat{X}_{s,t}}$.
\end{lem}
After application of the lemma to $Z_{s,t}$, the \Cref{eq:linear} becomes
\begin{equation}
     \hat{X}_{s,t} = \alpha  + \beta \hat{W}_{\hat{X}_{s,t}}, \quad a.s.
\end{equation}
The equation implies that if random variable takes value $\hat{X}_{s,t}= t$, then we have $\hat{W}_t = \frac{t - \alpha}{\beta}$. Therefore, we can write $\hat{X}_{s,t}$ as the first passage time of the Brownian motion $\hat{W}_t$:
\begin{equation}
    \hat{X}_{s,t} \stackrel{d}{=} \inf \left\{t \geq 0 : t  - \beta \hat{W}_t  = \alpha \right\}.
\end{equation}
The first passage time of a Brownian motion with positive drift of a level $\alpha$ is well-known to follow an inverse Gaussian distribution~\cite{folks1978inverse} $IG(\mu,\nu)$ with parameters $\mu = \alpha$, and $\gamma = \frac{\alpha^2}{\beta^2}$. The IG distribution has a probability density function $f_{IG(\mu,\gamma)}$ given by 
\begin{equation}
     f_{\hat{X}_{s,t}}(x)= f_{IG(\mu,\gamma)}(x) = \sqrt{\frac{\gamma}{2\pi x^3}} \exp\left(\frac{\gamma (x-\mu)^2}{2\mu^2x}\right), \quad x > 0.
\end{equation}
We can now define a simulation scheme: We define approximation variables $\hat{X}_{s,t}, \hat{Z}_{s,t}$, such that $\hat X_{s,t} = \alpha + \beta \hat Z_{s,t}$ almost surely, and then conclude that $\hat X_{s,t} \sim IG(\mu,\gamma)$. We can subsequently obtain samples of $\hat X_{s,t}$ by sampling from an IG distribution with given parameters, and derive $\hat Z_{s,t}$ from the given equation.  These samples are sufficient to forward-propagate the quantities $S_t$ through \Cref{eq:main_process}, using $ X_{s,t} \approx \hat X_{s,t}$ and $Z_{s,t} \approx \hat Z_{s,t}$. The samples for random variables $\hat X^n_{s,t}, n \leq N$ are then derived from $\hat Z_{s,t}$, also from a set of linear equations of the same form, a step which we will elaborate in \Cref{sec:proj}. First, we consider how to obtain a linear equation \Cref{eq:linear}, such that the random variables $\hat X_{s,t}$ and $\hat Z_{s,t}$ can be sampled. 
\subsection{A linear approximation between $X_{s,t}$ and $Z_{s,t}$}
Both $X_{s,t}$ and $Z_{s,t}$ are stochastic variables which share a dependency through the dynamics of the individual variance processes $U^n_t$, as seen from (\ref{eq:lifted_heston_end}). For our scheme, we require an approximate linear dependency between the variables, which we derive using a \emph{linear projection approach}. We obtain the linear equation by projecting the random variable $Z_{s,t}$ onto the linear subspace spanned by $X_{s,t}$ in $L^2$. The projection is derived from the conditional moments of $(X_{s,t}, Z_{s,t})$ on $\F_s$, which we obtain from the dynamics of the processes of the model, and we show how to obtain them in \Cref{lem:exp} and \Cref{lem:cov}. By definition the linear projection minimizes the $L^2$ error $\E[(\hat{X}_{s,t} - X_{s,t})^2 | \F_s]$ between the approximation $\hat{X}_{s,t} = \alpha + \beta \hat{Z}_{s,t}$ and the true variable $X_{s,t}$. The approach yields the \emph{best linear approximation} in a mean-squared sense (see \Cref{lem:optimal}), and avoids the approximation of the integral altogether.

\begin{rem}[Conditioning on $\F_s$]
    In the following results, we derive the conditional moments of the stochastic quantities required for the incremental changes as in (\ref{eq:main_process}),(\ref{eq:main_variance}), where the conditioning is always on the filtration $\F_s$ at the beginning of the increment. When it is clear, we will use the shorter notation $\E_s[X_{s,t}] := \E[X_{s,t} | \F_s]$ to denote the \emph{conditional expectation} on $\F_s$ for increased readability.
\end{rem}
The lemma below shows how to derive the linear projection given the moments of $(X_{s,t}, Z_{s,t})$.
\begin{lem}[Optimal Linear Projection]
\label{lem:optimal}
    Let $s<t$ be fixed and let $\hat{X}_{s,t}(\alpha,\beta) = \alpha + \beta Z_{s,t}$ be a linear predictor for the random variable $X_{s,t}$. The expression $\E_s\left[(X_{s,t} - \hat{X}_{s,t}(\alpha,\beta))^2\right]$ has a global minimum at $\alpha= \E_s[X_{s,t}] $ and $ \beta = \frac{\E_s[Z_{s,t}X_{s,t}]}{\E_s[Z_{s,t}^2]}$. 
\end{lem}
\begin{proof}
    The partial derivatives of $(a,b) \mapsto \E_s\left[(X_{s,t} - \hat{X}_{s,t}(\alpha,\beta))^2\right]$ are given by
    \begin{equation*}
        \frac{\partial }{\partial \alpha}\E_s\left[(X_{s,t} - \hat{X}_{s,t}(\alpha,\beta))^2\right] = 2\E_s[\alpha + \beta Z_{s,t} - X_{s,t}],
    \end{equation*}
    and
        \begin{equation*}
        \frac{\partial }{\partial \beta}\E_s\left[(X_{s,t} - \hat{X}_{s,t}(\alpha,\beta))^2\right] = 2\E_s[(\alpha + \beta Z_{s,t} - X_{s,t})Z_{s,t}].
    \end{equation*}
        Setting the derivatives equal to $0$ we obtain from the first equation that $\alpha = \E_s[X_{s,t}]$ since $\E_s[Z_{s,t}] = 0$, and from the second equation that \[\E_s[Z_{s,t}X_{s,t}] = \beta \E_s[Z_{s,t}^2] \iff \beta = \frac{\E_s[Z_{s,t}X_{s,t}]}{\E_s[Z_{s,t}^2]}.\]
    From the second-order derivatives, we confirm that the extreme point is a local minimum, and since the function is convex, the local minimum is a global minimum.
    \end{proof}
We obtain the optimal linear predictor for $X_{s,t}$ given $Z_{s,t}$ by evaluating the expectation of $\E_s[X_{s,t}],\E_s[Z^2_{s,t}]$ and the covariance $\E_s[Z_{s,t}X_{s,t}]$. Since the quadratic variation $\qv{Z}_{s,t}$ of the diffusion term $Z_{s,t}$ is $\qv{Z}_{s,t}= X_{s,t}$, we also have that $\E_s[X_{s,t}] =\E_s[Z^2_{s,t}]$ by It\^o's isometry, which leaves two quantities to be determined for the linear predictor. We will now show how to derive these quantities directly from the dynamics of the processes given by (\ref{eq:lifted_heston} - \ref{eq:lifted_heston_end}) in the subsequent two theorems. As a byproduct, we also derive the quantities $\E_s[X^n_{s,t}]$ and $\E_s[X^n_{s,t}Z_{s,t}]$ for all $n\leq N$, which will be use subsequently as well. The derivation involves computing the respective quantities on the state processes $U_t^n$ themselves, by solving $N$-dimensional linear systems, and then aggregating the processes to the dynamics of $V_t$. 
\begin{theorem}[Expectation of Integrated Variance]
\label{lem:exp}
    Consider the lifted Heston model given by equations (\ref{eq:lifted_heston} - \ref{eq:lifted_heston_end}) and consider a pair of times $s,t > 0$ such that $s<t$. Let $\bs{\omega} = (\omega_1, \omega_2, \dots, \omega_N)$,  $\bs{U}_s = (U_s^1,U_s^2, \dots, U_s^N)^T$ and the matrix $A$ be
    \begin{align}
    \nonumber A &= -\lambda \bs{1}_N\bs{\omega} - \text{diag}(\bs{x})=-\lambda 
\begin{pmatrix}
\omega_1 & \omega_2 & \cdots & \omega_N \\
\omega_1 & \omega_2 & \cdots & \omega_N \\
\vdots & \vdots & \ddots & \vdots \\
\omega_1 & \omega_2 & \cdots & \omega_N \\
\end{pmatrix}
-
\begin{pmatrix}
x_1 & 0 & \cdots & 0 \\
0 & x_2 & \cdots & 0 \\
\vdots & \vdots & \ddots & \vdots \\
0 & 0 & \cdots & x_N \\
\end{pmatrix},
\end{align}
with $\bs{1}_N = (1,1,\dots,1)^N$ the one-vector.
We can express the expectation of the integrated variance and state processes $X_{s,t}= \int_{s}^t V_{u} \d u$ and $X^n_{s,t} = \int_{s}^t U^n_{u} \d u$, such that
    \begin{align}
        \nonumber \E_s[X_{s,t}] &= \sum_{n=1}^N \omega_n \E_s[X^n_{s,t}] + \int_s^t g_0(u) \d u \\
        &= \bs{\omega} \left(A^{-1} \left(\e^{A(t-s)} - I_N\right)\boldsymbol{U}_s + \bs{\xi}_t\right)+ \int_s^t g_0(u) \d u,
    \end{align}
    where $I_N$ the identity matrix of size $N$ and $\bs{\xi}_t \in \R^N$ a solution to the initial value problem
        \begin{equation}
        \frac{\d \bs{\xi}_u}{\d u}  = A\bs{\xi}_u - \lambda G_0(s,u) \boldsymbol{1}_N,
    \end{equation} between $[s,t]$ with initial value $\bs{\xi}_s = \boldsymbol{0}_N := (0,0,\dots,0)^T$.
\end{theorem}
\begin{proof}
The exception $\E_s[X_{s,t}]$ is derived as the solution of a multidimensional initial value problem for the integrated state processes $X_{s,t}^n := \int_s^t U^n_u \d u$. The first step is to express the expectation of $X_{s,t}$ as a linear combination of $X^n_{s,t}$. From \Cref{eq:lifted_heston_middle}, we find
\[\E_s[X_{s,t}] = G_0(s,t) + \sum_{n=1} \omega_n \E_s[X_{s,t}^n],\]
where $G_0(s,t) := \int_s^t g_0(u) \d u$.
We will now derive a system of equations for $X^n_{s,t}$ for all $n \leq N$. Using the state processes dynamics of \Cref{eq:lifted_heston_end}, we derive that for $n \leq N$, we have
\begin{align*}
\E_s[X^n_{s,t}] &= (t-s)U_s^n + \int_s^t -x_n \E_s\left[ \int_s^u U^n_l \d l\right] -  \lambda \E_s\left[\int_s^u V_l \d l\right] \d u \\
&= (t-s)U_s^n  + \int_s^t -x_n \E_s\left[ X_{s,u}^n\right] -  \lambda \E_s\left[X_{s,u}\right] \d u \\
&= (t-s)U_s^n  +\int_s^t -x_n \E_s\left[ X_{s,u}^n\right] -  \lambda \left[G_0(s,u) + \sum_{m=1} \omega_m \E_s[X_{s,u}^m] \right] \d u. \end{align*}
Defining $\mu_t^n = \E_s[X_{s,t}^n]$ and $\boldsymbol{\mu}_t = (\mu_t^1, \mu_t^2,\dots,\mu_t^N)^T$, the array of expectations can be expressed as a linear system
\[\boldsymbol{\mu}_t = \int_{s}^t \boldsymbol{U}_s+ A \boldsymbol{\mu}_u - \lambda G_0(s,u)\boldsymbol{1}_N \d u,\]
or equivalently
\[\frac{\d \boldsymbol{\mu}_t}{\d t} = \boldsymbol{U}_s + A \boldsymbol{\mu}_t - \lambda G_0(s,t)\boldsymbol{1}_N. \]
We now split the system into two subproblems $\bs{\mu}_u = \bs \eta_u + \bs \xi_u, u \in [s,t]$, where $\bs{\eta}$ is stochastic containing $\boldsymbol{U}_s$, and $\bs{\xi}$ a deterministic term. Let $\bs{\eta}_t$ and $\bs{\xi}_t$ be such that\[\frac{\d \boldsymbol{\eta}_t}{\d t} = \boldsymbol{U}_s + A \boldsymbol{\eta}_t, \quad \frac{\d \boldsymbol{\xi}_t}{\d t} =  A \boldsymbol{\xi}_t - \lambda G_0(s,t)\boldsymbol{1}_N,\quad \bs\eta_s = \bs\xi_s = 0.\]

In this case, we find that $\boldsymbol{\eta}_t + \boldsymbol{\xi}_t$ solves the original system, since
\begin{align*}
\frac {\d}{ \dt} (\boldsymbol{\eta}_t + \boldsymbol{\xi}_t) &= \frac {\d}{ \dt}\boldsymbol{\eta}_t +\frac {\d}{ \dt} \boldsymbol{\xi}_t \\
&= \boldsymbol{U}_s + A \boldsymbol{\eta}_t + A \boldsymbol{\xi}_t - \lambda G_0(s,t)\boldsymbol{1}_N\\
&= \boldsymbol{U}_s + A (\boldsymbol{\eta}_t + \boldsymbol{\xi}_t )- \lambda G_0(s,t)\boldsymbol{1}_N.\end{align*}
We can therefore solve the subsystems separately and add their solutions. We begin with $\boldsymbol{\eta}_t$, which is an inhomogeneous system and its solution is
\[\boldsymbol{\eta}_t = \int_s^t \e^{A(t-u)} \du \, \boldsymbol{U}_s = A^{-1} \left(\e^{A(t-s)} - I_N\right)\boldsymbol{U}_s,\]
where $\e^A$ is the matrix exponential, i.e. the power series expansion $\sum_{n=0}^\infty \frac{A^n}{n!}$. Similarly, the second subproblem is also inhomogeneous, but since $G_0$ depends on $t$ as well, we have
\[\boldsymbol{\xi}_t = \int_s^t -\e^{A(t-u)} \lambda G_0(s,u)\boldsymbol{1}_N \du.\]
Since $\mu^n_t = \E_s[X_{s,t}^n]$, we find
\[\E_s[X_{s,t}] = \bs{\omega} \boldsymbol{\mu}_t + G_0(s,t) = \bs{\omega} \left(A^{-1} \left(\e^{A(t-s)} - I_N\right)\boldsymbol{U}_s + \bs{\xi}_t\right) + G_0(s,t).\]
\end{proof}
\begin{rem}[Numerically solving $\bs\xi$]
    To calculate the expectation of $X_{s,t}$, we require, in most cases\footnote{depends on the exact choice of $g_0$}, a numerical solver for the initial value problem of $\bs{\xi}_t$. These can be solved using standard libraries, such as \texttt{scipy.integrate.solve\_ivp}. Note that we could have also solved the differential equation $\bs{\mu}_t$ directly using such a numerical solver. This turns out to be infeasible, since the value then depends on $\bs{U}_s$, which differs path-wise, meaning that we need to solve an equation per path. By splitting the equation $\bs \mu_t = \bs \xi_t + \bs \eta_t$, we obtain a part $\bs \xi_t$ which requires a numerical technique, but which is independent of $\bs{U}_s$, and therefore needs to be calculated only once. The term $\bs \eta_t$, on the other hand, can be calculated analytically by a simple matrix multiplication. 
\end{rem}
Note that the expectation of the integrated variance could be expressed as an integration of $V_t$ in $t$ through \Cref{eq:expecation} by interchanging the integration and expectation. However, as the expression is a Volterra equation, it is available analytically and therefore cannot be integrated in a straightforward way. For this reason, we derive the expectation as shown in the proof. From the theorem we can derive well-known quantities of the Heston model of \Cref{eq:heston_simple}, such as the unconditional mean of variance.
\begin{cor}[Expectation of the Heston Model Variance]
    Let $V_t$ be the variance process in the Heston model and let $s < t$ be fixed times and let $t_0 = 0$. We have $A = -\lambda$ and $G_0(s,t) = V_0(t-s) +\frac{1}{2}\theta\lambda (t^2-s^2)$. The conditional mean of the integrated variance is then given by  \begin{align}
        \nonumber\E_s[X_{s,t}] &= -\frac{1}{\lambda}(\e^{-\lambda (t-s)})U^1_s + \int_{s}^t -\lambda \e^{-\lambda(t-u)} \left[V_0(u-s) +\frac{1}{2}\theta\lambda (u^2-s^2)\right] \d u \\
        &+V_0(t-s) +\frac{1}{2}\theta\lambda (t^2-s^2).
    \end{align}
    In particular, when $s = t_0 = 0$, we get the unconditional expectation $\E[X_t]$ as
    \begin{equation}
        \E[X_{0,t}] = -\frac{1}{\lambda} (V_0-\theta)\e^{-\lambda t} + \theta t + \frac{(V_0-\theta )}{\lambda}.
    \end{equation}
    Taking the derivative of $\E[X_{0,t}]$ we obtain the unconditional expectation of $V_t$ as
    \begin{equation}
    \frac{\d }{\dt}\E[X_{0,t}] = 
        \E[V_t] = (V_0- \theta)\e^{-\lambda t} + \theta.
    \end{equation}
\end{cor}
The second term for the linear projection we need to derive is the covariance term $\E_s[X_{s,t}Z_{s,t}]$ between the integrated variance and the diffusion term. We again reduce the problem first into its state-processes by computing $\E_s[X^n_{s,t}Z_{s,t}]$, from which we then aggregate the total covariance. \Cref{lem:cov} shows how to derive the expression. For its derivation, we require a supporting lemma regarding the integration of a specific term.
\begin{lem}
\label{lem:helper}
Let a matrix $A = -\lambda\bs{1}_N\bs{\omega} - \text{diag}(X)$ be a full-rank matrix of size $N$ and let $V,\Lambda$ be the eigen-decompositions of $A$, such that $A=V \Lambda V^{-1}$. We have that the integral $\int_{s}^t \e^{A (t-u)}\bs{1}_N \bs{\omega}\e^{A(u-s)} \du$ with $\bs{1}_N = (1,1,\dots,1)^T$ can be written as
\[ \int_{s}^t \e^{A (t-u)}\bs{1}_N\bs{\omega}\e^{A(u-s)} \du = V E V^{-1},\]
with $E \in \R^{N\times N}$ given by
\begin{align*}
    E_{i,j} &= u_iv_j \frac{\e^{(\lambda_j - \lambda_i)(t-s)} -1}{(\lambda_j - \lambda_i)},\quad\quad i,j \leq N \quad \text{ if } i\neq j, \\
    E_{i,i} &= \e^{\lambda_{i} (t-s)}(t-s),\quad\quad\quad\quad i \leq N,
\end{align*}
where $\bs{u},\bs{v} \in R^N$ are row vectors such that $ V (\bs{1}_N\bs{\omega})V^{-1} = \bs{u} \bs{v}^T$, and $\lambda_1,\lambda_2,\dots,\lambda_N$ are the distinct eigenvalues of $A$.
\end{lem}
\begin{proof}
First, we note that since $\bs{1}_N\bs{\omega}$ is of rank one, the matrix in the eigenbasis of $A$ is also of rank one, meaning that we can find vectors $\bs{u},\bs{v}$ using a singular value decomposition of $$Q := V\cdot \bs{1}_N\bs{\omega}\cdot V^{-1}.$$ We find matrices $L,M$ and a vector $\bs{s} = (s_1,0,\dots,0) \in \R^N$ with positive constant $s_1$, such that $Q = L \cdot\text{diag}(\bs{s}) \cdot M^T$. Then, we find $\bs{u},\bs{v}$ as the first row vectors of $L,M$, respectively, scaled by $\sqrt{s_1}$. We can then write the integral as
\[
\int_{s}^t \e^{A (t-u)}\bs{1}_N\bs{\omega}\e^{A(u-s)} \du = V \left[\int_{s}^t \e^{\Lambda (t-u)}\bs{u}\bs{v}^T\e^{\Lambda(u-s)} \du \right] V^{-1}.\]
Since $\Lambda$ is diagonal, so is $\e^{\Lambda}$, and the integrant reduces element-wise to the expression
\[ I_{i,j} = u_iv_j \e^{\lambda_i (t-u) - \lambda_j (u-s)}.\]
After element-wise integration from $s$ to $t$, we see that the integral of $I$ is $E$, and therefore conclude the proof of the lemma.
\end{proof}
The lemma will be used in the following result.
\begin{theorem}[Covariance of Integrated Variance and Diffusion]
    \label{lem:cov} 
    Consider the lifted Heston model given by equations (\ref{eq:lifted_heston} - \ref{eq:lifted_heston_end}) and let $A$ be a matrix as in \Cref{lem:exp}. Furthermore, let $A = V \Lambda V^{-1}$ be its eigen-decomposition. For a pair of times $s,t > 0$ such that $s<t$, we have that the covariance between $X_{s,t}$ and $Z_{s,t}$ is given by
    \begin{align}
\nonumber \E_s[X_{s,t}Z_{s,t}] = \sum_{n=1}^N\omega_n  \E_s[X^n_{s,t}Z_{s,t}] &=  \bs{\omega} \bigg[ \nu\left[ VEV^{-1} - A^{-1}(\e^{A(t-s)} -I)\bs{1}_N\bs{\omega} \right] A^{-1} \bs{U}_s \\
&+ \nu\int_{s}^t \e^{A(t-u)}  \left(\bs{1}_N\bs{\omega} \bs{\xi}_u + G_0(s,u)\right) \du\bigg],
    \end{align}
    where $\bs\xi_u$ solves the differential equation
            \begin{equation}
        \frac{\d \bs{\xi}_u}{\d u}  = A\bs{\xi}_u - \lambda G_0(s,u) \boldsymbol{1}_N,
    \end{equation} 
    with $\xi_s=0$, and $E \in \R^{N\times N}$ a matrix given by
    \begin{align*}
    E_{i,j} &= u_iv_j \frac{\e^{(\lambda_j - \lambda_i)(t-s)} -1}{(\lambda_j - \lambda_i)},\quad\quad i,j \leq N \text{ if } i\neq j, \\
    E_{i,i} &= \e^{\lambda_{i} (t-s)}(t-s),\quad\quad\quad\quad i \leq N,
\end{align*}
where $\bs{u},\bs{v} \in R^N$ are such that $ V (\bs{1}_N\bs{\omega})V^{-1} = \bs{u} \bs{v}^T$, and $\lambda_1,\lambda_2,\dots,\lambda_N$ are the distinct eigenvalues of $A$.
    \begin{proof}
        For the covariance, we will again utilize the integral representation of the state processes. We derive the covariance $\E_s[X^n_{s,t}Z_{s,t}]$ for any $n$. 
        \begin{align*}
            \E_s[X^n_{s,t}Z_{s,t}] &=
            \int_{s}^t\E_s[U^n_s Z_{s,t}]\du - \int_{s}^t x_n \E_s\left[\int_{s}^u U_l^n\d l Z_{s,t}\right]\d u \\ 
            &- \int_{s}^t \lambda \E_s\left[\int_{s}^u V_l\d l Z_{s,t}\right]\d u + \nu\int_s^t\E_s[Z_{s,u}Z_{s,t}] \du\\
            &= \E_s[U^n_s Z_{s,t}] (t-s) - \int_{s}^t x_n \E_s[X^n_u Z_{s,t}]\d u - \int_{s}^t \lambda \E_s[X_{s,u} Z_{s,t}]\d u + \nu \int_s^t\E_s[Z_{s,u}Z_{s,t}] \du.
        \end{align*}
        Since $Z_{s,t}$ is a martingale and independent of $\F_{s}$, the first term is equal to zero. Similarly, we can split each remaining integral into two by writing $Z_{s,t} = Z_{s,u} + Z_{u,t}$. Now the term $Z_{u,t}$ is centered and independent of $U_u^n$ and $Z_{s,u}$, meaning that we obtain
        \begin{align}
        \label{eq:intformxz}
             \nonumber \E_s[X^n_{s,t}Z_{s,t}] = -\int_{s}^t x_n \E_s[X^n_{s,u} Z_{s,u}]\d u - \int_{s}^t \lambda \E_s[X_{s,u} Z_{s,u}]\d u + \nu\int_s^t\E_s[X_{s,u}] \du\\
            = -\int_{s}^t x_n \E_s[X^n_{s,u} Z_{s,u}]\d u - \int_{s}^t \lambda \sum_{m=1}\omega_m\E_s[X^m_{s,u} Z_{s,u}]\d u + \nu\int_s^t \E_s[X_{s,u}] \du,
        \end{align}
        where we also used that $\E_s[Z^2_{s,u}] = \E_s[X_{s,u}]$. This is again a system of differential equations in integral form. Let $\kappa^n_t = \E_s[X^n_{s,t}Z_{s,t}]$ and $\bs{\kappa}_t = ( \kappa^1_t , \kappa^2_t, \dots , \kappa^N_t)^T$, and let $\eta_t, \xi_t$ be as in the notation of \Cref{lem:exp}, such that $\bs\eta_t + \bs\xi_t = \bs\mu_t$, with $\mu_t^n = \E_s[X^n_{s,t}]$. We can write the entire system after multiplication by $\bs{\omega}$ as the one-dimensional system
\begin{align*}
    \E_s[X_{s,t}Z_{s,t}] &= \bs{\omega}\bs\kappa_t \\
    &= \sum_{n=1}^N \omega_n \left[-\int_{s}^t x_n \kappa^n_u\d u - \int_{s}^t \lambda \sum_{m=1}\omega_m\kappa^m_u\d u + \nu\int_s^t \sum_{m=1}\omega_m \mu_u^n  + G_0(s,u)\du\right]\\
    &= \sum_{n=1}^N \omega_n \left[-\int_{s}^t x_n \kappa^n_u\d u - \int_{s}^t \lambda \sum_{m=1}\omega_m\kappa^m_u\d u + \nu\int_s^t \sum_{m=1}\omega_m (\xi_u^m  + \eta_u^m) + G_0(s,u)\du\right]\\
    &= \bs{\omega} \int_{s}^t A \bs \kappa_u + \nu\bs{1}_N \left[\bs{\omega}\bs \xi_u + \bs{\omega}\bs \eta_u + G_0(s,u)\right] \du.
\end{align*}
The system is again split into two separate equations $\bs\psi_t, \bs\chi_t$ such that $\bs\kappa_t =  \bs\chi_t + \bs\psi_t$. The term $\bs\psi_t$ contains the deterministic part $\bs{\xi}_u$ and $\bs\chi_t$ contains the stochastic term $\bs{\eta}_u$. We begin with the differential equations for $\bs\chi_t$:
\[\frac{\d}{\d t}\bs\chi_t = A\bs\chi_t + \nu\bs{1}_N \bs{\omega} \bs{\eta}_t.\]
Since $\bs{\eta}_t = A^{-1}\left[\e^{A(t-s)} - I_N\right]\bs{U}_s $ and the system is again inhomogeneous, we write
\begin{align*}
\bs\chi_t &= \nu\int_{s}^t \e^{A (t-u)}\bs{1}_N\bs{\omega}\left[\e^{A(u-s)} - I_N\right]A^{-1}\bs{U}_s \du,\\
&= \nu\left[\int_{s}^t \e^{A (t-u)}\bs{1}_N\bs{\omega}\e^{A(u-s)} \du - \int_{s}^t \e^{A (t-u)}\bs{1}_N \bs{\omega}\du\right]A^{-1}\bs{U}_s ,
\end{align*}
since $A^{-1}$ commutes with $A$ and therefore with its exponential power series. While the second integral is easily analytically solvable, the first integral is more difficult. The matrix $\bs{1}_N\bs{\omega}$ does not commute with $A$ or its power series, and we cannot simply move it outside of the expression. We use the relation of $A = \lambda\bs{1}_N\bs{\omega} - \text{diag}(X)$ to solve the integral, which we derive in \Cref{lem:helper}. The second integral is equal to $A^{-1}(\e^{A(t-s)} -I)\bs{1}_N\bs{\omega} $. We conclude that
\[\bs\chi_t = \nu\left[ VEV^{-1} - A^{-1}(\e^{A(t-s)} -I)\bs{1}_N\bs{\omega} \right] A^{-1} \bs{U}_s.\]

 The second equation $\bs{\psi}_t$ is deterministic and solves the differential equation
\[\frac{\d}{\d t}\bs{\psi}_t = A \bs\psi_t + \nu (\bs{1}_N\bs{\omega} \bs{\xi}_t + G_0(s,t)),\]
which has a solution
\[\bs\psi_t = \nu\int_{s}^t \e^{A(t-u)}  \left(\bs{1}_N\bs{\omega} \bs{\xi}_u + G_0(s,u)\right) \du.\]
We conclude that since $\bs\kappa_t = \bs\psi_t + \bs \chi_t $, we have therefore that
\begin{align*}
    \E_s[X_{s,t}Z_{s,t}] &= \bs{\omega} \kappa_t = \bs{\omega} (\bs\chi_t + \bs\psi_t) \\
    &=\bs{\omega}\bigg[ \nu\left[ VEV^{-1} - A^{-1}(\e^{A(t-s)} -I)\bs{1}_N\bs{\omega} \right] A^{-1} \bs{U}_s + \nu\int_{s}^t \e^{A(t-u)}  \left(\bs{1}_N\bs{\omega} \bs{\xi}_u + G_0(s,u)\right) \du\bigg].
\end{align*}
    \end{proof}
    \end{theorem}
\subsection{A Linear Projection Algorithm}
\label{sec:proj}
We combine the theoretical results of \Cref{lem:exp} and \Cref{lem:cov} and the insights from \Cref{sec:explicit_sol}, in particular \Cref{eq:main_variance}, to derive the (constrained) linear projection (C-LP) algorithm to simulate paths of the lifted Heston model. We will first lay the basis of the algorithm, which is purely based on the linear projection. In a second step, we then introduce a constraint of the algorithm to ensure that it is well-defined at all times.

Suppose we define a time discretization $t_0 < t_1 < t_2 \dots < t_M$ and are currently at a time $t_i$ with $0 \leq i \leq M$. Given the current state of $\bs{U}_{t_i}$, we compute the (conditional) linear projection of $Z_{t_i,t_{i+1}}$ onto the linear subspace of $X_{t_i,t_{i+1}}$. We define: 
\begin{equation}
    \alpha_i = \E_{t_i}[X_{t_i,t_{i+1}}], \quad \beta_i = \frac{1}{\alpha_i}{\E_{t_i}[X_{t_i,t_{i+1}}Z_{t_i,t_{i+1}}]},
\end{equation}
using the theorems from the previous section. Note that $\alpha_i$ and $\beta_i$ are conditional expectations, and hence path-wise distinct. The linear approximation defines the relation $\hat X_{t_i,t_{i+1}} = \alpha_i + \beta_i\hat Z_{t_i,t_{i+1}}$, where $\hat X_{t_i,t_{i+1}}$ then follows an inverse Gaussian distribution $\hat X_{t_i,t_{i+1}} \sim IG(\mu,\gamma)$ with coefficients $\mu = \alpha$, $\gamma = \frac{\alpha^2}{\beta^2}$. We sample (path-wise) from such a distribution to obtain samples of $\hat X_{t_i,t_{i+1}}$, and samples for $\hat{Z}_{t_i,t_{i+1}}$ are then calculated as $\hat{Z}_{t_i,t_{i+1}} = \frac{\hat X_{t_i,t_{i+1}} - \alpha_i}{\beta_i}$. To update the state processes $\bs{U}_{t_i}$ to $\bs{U}_{t_{t+1}}$, we also require $\hat{X}^n_{t_i,t_{i+1}}$ for all $n\leq N$. Theorems \ref{lem:exp} and \ref{lem:cov} also provide linear projection coefficients for $\E_{t_i}[X^n_{t_i,t_{i+1}}]$ and $\E_{t_i}[X^n_{t_i,t_{i+1}}Z_{t_i,t_{i+1}} ]$, from which we obtain
\begin{equation}
    \alpha^n_i = \E_{t_i}[X^n_{t_i,t_{i+1}}], \quad \beta^n_i = \frac{1}{\alpha_i}{\E_{t_i}[X^n_{t_i,t_{i+1}}Z_{t_i,t_{i+1}}]}, \quad n \leq N.
\end{equation}
This means that we obtain the samples \begin{equation}
    \hat{X}^n_{t_i,t_{i+1}} = \alpha^n_i + \beta^n_i \hat{Z}_{t_i,t_{i+1}}.\end{equation}
    The samples for $\hat{X}^{n}_{t_i,t_{i+1}}$ are consistent with $\hat{X}_{t_i,t_{i+1}}$, as we have that 
\begin{equation}
    \hat{X}_{t_i,t_{i+1}} = \sum_{n=1}^N \omega_n \hat{X}^{n}_{t_i,t_{i+1}} + \int_{t_i}^{t_{i+1}}g_0(u)\du.
\end{equation}
We can now update the state processes as
\begin{equation}
\label{eq:update_u}
    \hat{U}^n_{t_{i+1}} = \hat{U}^n_{t_{i}} - x_n \hat X^n_{t_i,t_{i+1}} - \lambda \hat X_{t_i,t_{i+1}} + \nu \hat Z_{t_i,t_{i+1}},
\end{equation}
and derive $\hat{V}_{i+1}$ for the new states from \Cref{eq:main_variance} (or directly as the weighted sum of $U_{t_{i+1}}^n$), as well as the new values for the asset price from \Cref{eq:main_process}.

\subsection{Constraining the Slope}
While the linear projection method yields the $ L^2$-optimal relation between the variables, it is nevertheless an approximation of a nonlinear relation. For this reason we have to ensure that a potential approximation scheme $\left\{(\hat{X}_{t_i,t_{i+1}},\hat{Z}_{t_i,t_{i+1}}) :  0 \leq  i \leq M-1\right\}$ of the process is well-defined at all times. While deriving the coefficients $\alpha_i,\beta_i$ is not problematic, we have to ensure that the following conditions are fulfilled: 
$$\alpha_i >0  \text{ and } \beta_i \neq 0, \quad 0 \leq  i \leq M-1.$$ The conditions ensure that the inverse Gaussian distribution is well-defined at all time steps, and that $\hat Z_{t_i,t_{i+1}}$ can be computed. To ensure these properties, we will systematically restrict the parameter choice for $\beta_i$ on a \emph{feasibility domain}, where it is ensured that $\beta_i \neq 0$ and the subsequent step leads to a parameter $\alpha_{i+1} > 0$. Since $\alpha_{i+1}$ is the expectation of $\int_{t_{i+1}}^{t_{i+2}} V_u \du$, conditioned on $\hat{V}_{t_{i+1}}$, the restriction is obtained by ensuring that $\hat{V}_{t_{i+1}} = \bs{\omega} \hat{\bs {U}}_{t+1} + g_0(t_{i+1})$ is positive at all times. To ensure that this quantity remains positive, we define a set of constraints with the current state for $\hat{\bs {U}}_{t}$ for the parameter choice of $\beta_i$, and the updated values. We express the value for the updated variance using (\ref{eq:update_u}), i.e. $\hat{V}_{t_{i+1}} = \bs\omega \bs {\hat U}_{t_{i+1}}+ g_0(t_{i+1})$, as a function of the value of $x=\hat{X}_{t_i,t_{i+1}}$, and the coefficient $\beta$:
\begin{equation}
\label{eq:constraint}
    C_i(x, \beta)  =  \bs{\omega} \hat{\bs {U}}_{t_i} + \bs{\omega} (-\text{diag}({\bs{x}})\hat{\bs{X}}_{t_i,t_{i+1}}(x) - \lambda x + \nu \frac{x-\alpha_i}{\beta}  )+ g_0(t_{i+1}),
\end{equation}
with $\hat{\bs{X}}_{t_i,t_{i+1}}(x) = \left(\hat{X}^1_{t_i,t_{i+1}}(x),\hat{X}^2_{t_i,t_{i+1}}(x),\dots,\hat{X}^N_{t_i,t_{i+1}}(x)\right)^T$ and $    \hat{X}^n_{t_i,t_{i+1}}(x) = \alpha^n_i 
        + \frac{ \beta^n_i (x - \alpha_i)}{\beta_i}$.
Note that the projection for $\hat{X}^n_{t_i,t_{i+1}}$ remains unchanged, since the state processes $\hat{U}^n_{t_{i+1}}$ can take negative values while still leading to positive $\hat{V}_{t_{i+1}}$, and hence require no constraining. The function $C(x, \beta)$ defines a set of open constraints $\left\{C_i(x, \beta) > 0 : x\in (0,\infty) \right\}$. Since the set contains uncountably many constraints, we aim to replace the set with an equivalent set of constraints that is much smaller. It turns out that the set can be reduced to a single closed constraint $C_i(0,\beta) \geq 0$ if we restrict $\beta$ to a smaller domain.
\begin{lem}
    Let $\left\{C_i(x, \beta) > 0 : x \in [0,\infty)\right\}$ be a collection of constraints given by \Cref{eq:constraint}. Denote the scalar $\bar{\omega} := \bs \omega \bs{1}_N  =\sum_{n=1}^N \omega_n$. The coefficient $\beta \in (0,\infty)$ satisfies 
    \begin{equation}
        \beta \leq \beta^{L_i}:=\frac{\nu\bar{\omega}}{\sum_{n=1}^N \omega_n \left(x_n\frac{\E_{t_i}[X^n_{t_i,t_{i+1}}Z_{t_i,t_{i+1}}]}{\E_{t_i}[X_{t_i,t_{i+1}}Z_{t_i,t_{i+1}}]} + \lambda \right)},
    \end{equation}
    and $C_i(0, \beta) \geq 0$ if and only if all constraints in $\left\{C_i(x, \beta) > 0 : x \in [0,\infty)\right\}$ are satisfied. Furthermore, the linear projection estimate for $\beta_i$ is smaller than $\beta^{L_i}$.
\end{lem}
\begin{proof}
    We note that $x \mapsto C_i(x, \beta)$ is a linear function for fixed $\beta$. This means that the constraint at $0$ fulfills all constraints for a $\beta$ if and only if the slope is positive, meaning that $\frac{\d C_i(x, \beta) }{\d x} \geq 0$. If the slope is negative, there is an $x > 0$ such that $C(x,\beta) < 0$. The derivative can be found by differentiating \Cref{eq:constraint}:
    \begin{align*}\frac{\d C_i(x, \beta) }{\d x} =& \sum_{n=1}^N \omega_n \left(-x_n\frac{\E_{t_i}[X^n_{t_i,t_{i+1}}Z_{t_i,t_{i+1}}]}{\E_{t_i}[X_{t_i,t_{i+1}}Z_{t_i,t_{i+1}}]} - \lambda + \frac{\nu }{\beta}\right).
    \end{align*}
    Since $\nu > 0$, the derivative is non-negative for all $\beta \in (0,\beta^{L_i}]$ if and only if it is the case for $\beta = \beta^{L_i}$. The upper boundary $\beta^{L_i}$ is chosen such that the derivative is $0$. 
    
    To prove that $\beta_i \leq \beta^{L_i}$, we show that the derivative is positive at $\beta_i$. We have 
    \begin{align*}
    \frac{\d C_i(x, \beta_i) }{\d x}=&\sum_{n=1}^N \omega_n \left(-x_n\frac{\E_{t_i}[X^n_{t_i,t_{i+1}}Z_{t_i,t_{i+1}}]}{\E_{t_i}[X_{t_i,t_{i+1}}Z_{t_i,t_{i+1}}]} - \lambda + \frac{\nu \E_{t_i}[X_{t_i,t_{i+1}}] }{\E_{t_i}[X_{t_i,t_{i+1}}Z_{t_i,t_{i+1}}]}\right)\\
    =&\frac{1}{\E_{t_i}[X_{t_i,t_{i+1}}Z_{t_i,t_{i+1}}]}\sum_{n=1}^N \omega_n  \left(-x_n\E_{t_i}[X^n_{t_i,t_{i+1}}Z_{t_i,t_{i+1}}] - \lambda \E_{t_i}[X_{t_i,t_{i+1}}Z_{t_i,t_{i+1}}] + \nu \E_{t_i}[X_{t_i,t_{i+1}}]\right).
    \end{align*}
    The last expression is similar to the calculations of \Cref{lem:cov}. We can show that 
    \begin{align*}
        \frac{\d}{\dt}\E_{t_i}[X^n_{t_{i}, t}Z_{t_{i}, t}] \bigg|_{t=t_{i+1}} = -x_n\E_{t_i}[X^n_{t_i,t_{i+1}}Z_{t_i,t_{i+1}}] - \lambda \E_{t_i}[X_{t_i,t_{i+1}}Z_{t_i,t_{i+1}}] + \nu \E_{t_i}[X_{t_i,t_{i+1}}],
    \end{align*}
    by differentiation the \Cref{eq:intformxz}. We will hence show that $$\E_{t_i}[X_{t_i,t_{i+1}}Z_{t_i,t_{i+1}}], \quad \text{ and }\quad \frac{\d}{\dt}\E_{t_i}[X_{t_{i}, t}Z_{t_{i}, t}] \bigg|_{t=t_{i+1}} = \frac{\d}{\dt}\sum_{n=1}^N \omega_n \E_{t_i}[X^n_{t_{i}, t}Z_{t_{i}, t}] \bigg|_{t=t_{i+1}},$$
    are both positive, in which case the slope is indeed positive. We prove this from the differential equation of $\E_{t_i}[X_{t_i,t_{i+1}}Z_{t_i,t_{i+1}}]$. Again, from the proof of \Cref{lem:cov}, we know that

    \[\frac{\d}{\dt} \bs \kappa_t = A \bs \kappa_t + \nu \bs{1}_N \mu_t ,\]
    where $\kappa^n_t = \E_{t_i}[X^n_{t_i,t}Z_{t_i,t}]$ and $\mu_t = \E_{t_i}[X_{t_i,t}]$ for $t \geq t_{i}$. The system has a solution
    \[\bs \kappa_t = \nu\int_{t_{i}}^{t} \e ^{A(t - u)}\bs{1}_N \mu_u \d u.\] To obtain $\E_{t_i}[X_{t_i,t}Z_{t_i,t}]$ we multiply with $\bs \omega$:
    \[\E_{t_i}[X_{t_i,t}Z_{t_i,t}] =\nu \bs\omega\int_{t_{i}}^{t} \e ^{A(t - u)}\bs{1}_N \mu_u \d u = \nu \int_{t_{i}}^{t} \bs\omega \e ^{A(t - u)}\bs{1}_N \mu_u \d u.\]
    We will now prove that the integrant is non-negative for all $u \in [t_i,t] $. This implies that $\E_{t_i}[X_{t_i,t}Z_{t_i,t}]$ is non-negative and increasing, hence has a non-negative derivative. Since $\mu_t$ is non-negative, as $V_t$ is almost surely non-negative, it suffices to show that $\bs\omega \e ^{A(t - u)}\bs{1}_N \geq 0$ for all $u \in [t_i,t].$

    The exponential of the matrix $A = -\lambda \bs{1}_N\bs\omega - \text{diag}(x)$ can be written as \[\e ^{A (t-u)} = \e^{-\lambda \bs{1}_N\bs\omega}\e^{- \text{diag}(x) (t-u)}, \] since the terms commute. The term $\e^{- \text{diag}(x) (t-u)}$ is the exponential of a diagonal matrix, and hence has only positive entries on the diagonals. Multiplying with $\bs{1}_N$ then yields
    \[\e^{- \text{diag}(x) (t-u)}\bs{1}_N = 
        (\e^{-x_1(t-u)},
        \e^{-x_2(t-u)},
        \dots,
        \e^{-x_N(t-u)}
    )^T, \quad u \in [t_i,t],\] which is a vector of only non-negative entries.
    
    We now consider the other term $\bs{\omega}\e^{-\lambda \bs{1}_N \bs\omega}$. We calculate the $m$ power $(\bs{1}_N\bs\omega)^m$ for any integer $m\geq0$. Then, $(\bs{1}_N\bs\omega)^m = \bs{1}_N\bar{\omega}^{m-1}\bs\omega$. This means that we can write the exponential as
    \[\bs{\omega}\e^{-\lambda \bs{1}_N \bs\omega} = \bs\omega \sum_{m=0}^{\infty} \frac{(-\lambda)^m(\bs{1}_N\bs\omega)^m}{m!} = \bs\omega\sum_{m=0}^{\infty} \frac{(-\lambda)^m\bs{1}_N\bar{\omega}^{m-1}\bs\omega}{m!} 
    =\sum_{m=0}^{\infty} \frac{(-\lambda)^m\bar{\omega}^{m}\bs\omega}{m!}, \]
    which is exactly equal to $\e^{-\lambda \bar{\omega}}\bs\omega$. This is again a (row) vector with only positive entries, since $\omega_n \geq 0, n \leq N$. The inner product of the two non-negative vectors is therefore a non-negative scalar, which concludes the proof that $\bs\omega \e ^{A(t - u)}\bs{1}_N \mu_u$ is non-negative for all $u \in [t_i,t]$, and therefore proves the claim.
\end{proof}
The conclusion from the lemma is that it suffices to consider a single closed constraint to ensure that $\hat{V}_{t_{i+1}}$ stays positive. This constraint helps us define a definition of feasibility.
\begin{dfn}
    Let $\bs U_{t_i}$ be a given state of the lifted Heston model at time $t_i$, and $t_{i+1} > t_i$ a second time for a time step $t_{i+1} - t_i$. We call the coefficient $\beta$ \emph{feasible} if $\beta \in (0,\beta^{L_i}]$ and $ C_i(0, \beta) \geq 0$.
\end{dfn}
With the definition of feasibility, we can define the constrained algorithm by obtaining the parameters $\alpha_i,\beta^C_i$ as the $L^2$ minimizer, under the constraint $C_i(0,\beta)$. The coefficient $\beta^C_i$ is then found as
\begin{equation}
    \beta^C_i \;:=\; \arginf_{\beta \in (0,\beta^{L_i}] } 
        \quad  \E_{t_i}\!\left[\left(X_{t_i,t_{i+1}} - (\alpha_i + \beta Z_{t_i,t_{i+1}})\right)^2\right]  \quad \text{s.t.}\quad  C_i(0,\beta) \geq 0.
\end{equation}
The minimization problem is well-defined if the constrained set for $\beta$ is non-empty.
\begin{lem}
    Let $\bs U_{t_i}$ be a given state of the lifted Heston model at time $t_i$, and $t_{i+1} > t_i$ a second time for a time step $t_{i+1} - t_i$. Furthermore, let $\alpha_i = \E_{t_i}[X_{t_i,t_{i+1}}]$ and $\beta_i = \frac{\E_{t_i}[X_{t_i,t_{i+1}}Z_{t_i,t_{i+1}}]}{\alpha_i}$. The set $\mathcal{B}$ given by
    \begin{equation}
        \mathcal{B} = \{ \beta \in (0,\infty):  \beta  \text{ is feasible } \},
    \end{equation}
    is non-empty. 
\end{lem}
\begin{proof}
It suffices to show that $\beta^{L_i}$ is feasible, and therefore $C_i(0,\beta^{L_i}) \geq 0$. We find
\begin{align}
\label{eq:beta_upper}
   \nonumber C_i(0,\beta^{L_i}) &= \bs{\omega} \hat{\bs {U}}_{t_i} + \bs{\omega} (-\text{diag}({\bs{x}})\hat{\bs{X}}_{t_i,t_{i+1}}(0) - \nu \frac{\alpha_i}{\beta^{L_i}}  )+ g_0(t_{i+1})\\
    \nonumber&= \sum_{n=1}^N \omega_n \hat{U}^n_{t_i} + \sum_{n=1}^N \omega_n \left[-x_n \left(\alpha^n_i -\frac{\E_{t_i}[X^n_{t_i,t_{i+1}}Z_{t_i,t_{i+1}}]}{\E_{t_i}[X_{t_i,t_{i+1}}Z_{t_i,t_{i+1}}]} \alpha_i \right)\right]\\
    \nonumber&- \nu \frac{\bar{\omega}\alpha_i {\sum_{n=1}^N \omega_n \left(x_n\frac{\E_{t_i}[X^n_{t_i,t_{i+1}}Z_{t_i,t_{i+1}}]}{\E_{t_i}[X_{t_i,t_{i+1}}Z_{t_i,t_{i+1}}]} + \lambda \right)}}{\nu \bar{\omega}}  + g_0(t_{i+1})\\
    &= \sum_{n=1}^N \omega_n \hat{U}^n_{t_i}  + \sum_{n=1}^N \omega_n (-x_n \alpha_i^n + -\lambda \alpha_i) + g_0(t_{i+1}).
\end{align}
The last expression is exactly equal to the expectation $\E_{t_i}[V_{t_{i+1}}]$, conditioned on $\F_{t_i}$, and is therefore positive.
\end{proof}
The proof shows that there exists a suitable $\beta^C_i$ for the optimization. The next result shows how to obtain it analytically
\begin{lem} Let $c :=\bs{\omega} \hat{\bs {U}}_{t_i} - \bs{\omega} (\text{diag}({\bs{x}}) \hat{\bs{X}}_{t_i,t_{i+1}}(0)  )+ g_0(t_{i+1})$
    The optimal slope $\beta^C_i$ is given by
    \begin{equation}
        \beta^C_i =\begin{cases}
            \beta_i \quad &\text{if } {\beta_i} \text{ is feasible},\\
           \frac{\nu \alpha_i \bar{\omega}}{c} \quad  &\text{if } {\beta_i} \text{ is not feasible.}
        \end{cases}
    \end{equation}
\end{lem}
\begin{proof}
    Firstly, note that $c$ is strictly larger than $0$, since it is greater than or equal to $\cref{eq:beta_upper}$. If the unconstrained minimum is feasible, it is the constrained minimum. Hence, $\beta_i = \beta_i^C$ in this case. Since $\alpha_i \geq 0$, the constraint is convex, and we can solve the optimization problem using the Karush-Kuhn-Tucker conditions. The conditions state that the minimum is on the boundary $C_i(0, \beta^C_i) = 0$, if $\beta_i$ is not feasible. Solving for $C_i(0, \beta^C_i) = 0$ for $\beta^C_i$, we obtain the solutions for the cases. Lastly, by shuffling the terms we see that $\frac{\nu \alpha_i \bar{\omega}}{c} \leq \beta^{L_i}$, so that $C_i(0, \beta^C_i) = 0$ is attainable.
\end{proof}
We summarize the constrained linear projection (C-LP) algorithm in algorithmic form. A Python implementation of the algorithm can be found on \href{https://github.com/NFZaugg/C-LP-Lifted-Heston}{Github}.
\begin{algorithm}[H]
\caption{Constrained Linear Projection (C-LP) Algorithm}
\label{alg:constrained}
{
\begin{algorithmic}[1]
\Require $S_0, t \mapsto g_0(t)$ \Comment{Starting values of $S$ and initial variance curve}
\Require $0 = t_0 < t_1 < t_2, \dots, t_M = T : \Delta t= t_{i+1} - t_i$ \Comment{An equidistant time discretization}
\Require $N, \lambda, \nu, \sigma, \rho, r, \bs{x} = (x_1, x_2, \dots, x_N), \bs{\omega} = (\omega_1, \omega_2, \dots, \omega_N)$ \Comment{Model Parameters}
\State Set $\bs{u} = (0, 0, \dots, 0)^T,\quad s = S_0$
\For{$i \gets 0$ to $M-1$}
\State  $G_0(t_i,t_{i+1}) = \int_{t_i}^{t_{i+1}} g_0(t) , dt$ \Comment{Calculate analytically if possible, otherwise numerically.}
    \State
        $\alpha_i = \E_{t_i}[X_{t_i,t_{i+1}}], \quad 
        \beta_i = \frac{1}{\alpha_i}\E_{t_i}[X_{t_i,t_{i+1}}Z_{t_i,t_{i+1}}].$ \hfill \Comment{Derive from \cref{lem:exp,lem:cov}}
    
    \If{$C_i(0,\beta_i) \geq 0$}
    \State $\beta_i^C = \beta_i$\Comment{Use unconstrained algorithm}
    \Else
            \State $
\hat{X}^n_{t_i,t_{i+1}}(0)
        = \E_{t_i}[X^n_{t_i,t_{i+1}}]  
        - \frac{\E_{t_i}[X^n_{t_i,t_{i+1}}Z_{t_i,t_{i+1}}]}{\E_{t_i}[X_{t_i,t_{i+1}}Z_{t_i,t_{i+1}}]} 
          {\E_{t_i}[X_{t_i,t_{i+1}}]},
        \quad n \leq N.
    $
    \State
    $\beta_i^C = \frac{\nu \alpha_i \sum_{n=1}^N \omega_n}{\bs{\omega} u - \bs{\omega} (\text{diag}({\bs{x}}) \hat{\bs{X}}_{t_i,t_{i+1}}(0)  )+ g_0(t_{i+1}) }.$ 
 \Comment{Constrain the beta}

    \EndIf
    \State Sample $\hat{X}_{t_i,t_{i+1}}$ from an inverse Gaussian distribution with 
$ 
        \mu = \alpha_i, 
        \gamma = \left(\frac{\alpha_i}{\beta^C_i}\right)^2.
$
    \State
    $
        \hat{Z}_{t_i,t_{i+1}} = \frac{1}{\beta^C_i}{(\hat{X}_{t_i,t_{i+1}} - \alpha_i)}.
    $ \Comment{ Compute projected diffusion}
        \State
    $\hat{X}^n_{t_i,t_{i+1}} 
        = \E_{t_i}[X^n_{t_i,t_{i+1}}]  
        + \frac{\E_{t_i}[X^n_{t_i,t_{i+1}}Z_{t_i,t_{i+1}}]}{\E_{t_i}[X_{t_i,t_{i+1}}Z_{t_i,t_{i+1}}]} 
         \left(\hat{X}_{t_i,t_{i+1}}-\E_{t_i}[X_{t_i,t_{i+1}}]\right), 
        \quad n \leq N.$
    
    \State $\bs{\hat{X}}_{t_i,t_{i+1}} = (\hat{X}^1_{t_i,t_{i+1}}, \hat{X}^2_{t_i,t_{i+1}}, \dots, \hat{X}^N_{t_i,t_{i+1}})^T$. \Comment{ Unconstrained projection for states}
    \State 
        $\bs{u} \gets \bs{u} - \text{diag}({\bs{x}}) \, \bs{\hat{X}}_{t_i,t_{i+1}} - \lambda \, \hat{X}_{t_i,t_{i+1}} + \sigma \, \hat{Z}_{t_i,t_{i+1}}$.
    \Comment{Update $\bs u$}
   \State 
        $\hat{V}_{t_{i+1}} = \bs{\omega} \, \bs{u} + g_0(t_{i+1})$.
    \Comment{Compute new  $\hat{V}_{t_{i+1}}$}
    \State Finally, update the price process for $s$ with
    \[
        s \gets s \, \e^{r \, \Delta t - \frac{1}{2} \hat{X}_{t_i,t_{i+1}} + \rho \, \hat{Z}_{t_i,t_{i+1}} + \sqrt{(1-\rho^2)\hat{X}_{t_i,t_{i+1}} } \, N_{t_i,t_{i+1}}},
    \]
    {where $N_{t_i,t_{i+1}}$ is a standard normal random variable, $N_{t_i,t_{i+1}} \sim \mathcal{N}(0, 1)$.}
\EndFor
\Return $s$
\end{algorithmic}
}
\end{algorithm}
\subsection{Weak Consistency}
In this section, we analyze the properties of convergence of the numerical scheme. Since the sampled quantity is not a Gaussian noise, showing a strong convergence property is challenging, as the discretization process must be reformulated in a time-continuous It\^o process. As Kloeden and Eckhard~\cite{kloeden2013numerical} point out, an easier property to verify is \emph{weak consistency}, which says that the simulation increments in the limit match the moments of the true increments. 
\begin{lem}[Weak Consistency]
    Let $\hat{V}^{h}_t$ denote a C-LP simulation as described in \Cref{alg:constrained} of the lifted Heston model with step size $h > 0$. Then, for any $t \geq t_0$, the following limits are fulfilled:
    \begin{align}
        &\lim_{h \to 0} \frac{\E_t\left[\hat{V}^{h}_{t+h} - \hat{V}^{h}_t  \right]}{h} = \lim_{h \to 0} \frac{\E_t\left[{V}^{h}_{t+h} - {V}^{h}_t  \right]}{h} = g_0'(t) + \sum_{n=1}^N \omega_n \left(-x_n U_t^n - \lambda V_t\right),\\
        &\lim_{h \to 0} \frac{\var_t\left[\hat{V}^{h}_{t+h} - \hat{V}^{h}_t \right]}{h} = \lim_{h \to 0} \frac{\var_t\left[{V}^{h}_{t+h} - {V}^{h}_t \right]}{h} = \nu^2V_t \sum_{n=1}^N \omega_n^2.
    \end{align}
    \begin{proof}
     Let $\alpha,\beta$ be the mean and, possibly constrained, slope parameter at time $t$ for stepsize $h$. The increment $\Delta^h \hat{V}_t := \hat{V}^{h}_{t+h} - \hat{V}^{h}_t$ is the one step update, and therefore given by
    \[\Delta^h \hat{V}_t = g_0(t+h) - g_0(t) + \sum_{n=1}^N \omega_n \left(-x_n \hat{X}^n_{t,t_h} -\lambda \hat{X}_{t,t+h} + \nu \frac{\hat{X}_{t,t+h} - \alpha}{\beta}\right), \]
    where $\hat{X}_{t,t+h}$ is an inverse Gaussian random variables with $\mu = \alpha, \gamma = \frac{\alpha^2}{\beta^2}$, with $\alpha = \E_{t}[X_{t,t+h}]$.
    Since the expectation of an IG random variable is $\mu$, we see that the diffusion term $ \frac{\hat{X}_{t,t+h} - \alpha}{\beta}$ falls off in expectation. The random variable $\hat{X}^n_{t,t_h}$ is given by 
    \[        \hat{X}^n_{t,{t+h}} 
        = \E_{t}[X^n_{t,{t+h}}]  
        + \frac{ \E_{t}[X^n_{t,{t+h}}Z_{t,{t+h}}]}{\E_{t}[X_{t,{t+h}}Z_{t,{t+h}}]} 
         \left(\hat{X}_{t,{t+h}}-\E_{t}[X_{t,{t+h}}]\right), 
        \quad n \leq N.\]
Again, in expectation, the second term falls off as $\E_{t}\left(\hat{X}_{t,{t+h}}-\E_{t}[X_{t,{t+h}}]\right) = 0$. We conclude that
\[\E_{t}[\Delta^h \hat{V}_t] = \Delta g_0(t) + \sum_{n=1}^N \omega_n \left(-x_n \alpha^n -\lambda \alpha \right),\]
with $\alpha^n := \E_{t}[X^n_{t,{t+h}}]$. Since $X_{t,t_h}$ is the integral of $V_t$, we see that $\lim_{h \to 0} \frac{\alpha}{h} = V_t$, and equivalently, $\lim_{h \to 0} \frac{\alpha^n}{h} = U^n_t$. This proves the first limit. 

For the variance, we note that $\var_{t}(\hat{X}_{t,t+h}) = \frac{\mu^3}{\gamma} = \alpha \beta^2$, and therefore $$\var_{t}(\lambda \hat{X}_{t,t+h}) =  \lambda^2\alpha\beta^2, \quad \var_{t}(\nu \frac{\hat{X}_{t,t+h} - \alpha}{\beta}) = \nu^2\alpha,$$
and similarly,
\[ \var_{t}(\hat{X}^n_{t,t+h} ) = \left(\frac{ \E_{t}[X^n_{t,{t+h}}Z_{t,{t+h}}]}{ \E_{t}[X_{t,{t+h}}Z_{t,{t+h}}]}\right) ^2 \alpha \beta^2.\]
We will now show that $|\beta| \to 0$, as $h \to 0$. Assume first that $\beta$ is not constrained. Using Cauchy-Schwartz, we have
\[|\beta |= \frac{\left|\E_{t}[X_{t,t+h}Z_{t,t+h}]\right|}{\E_{t}[X_{t,t+h}]} \leq \frac{\E_{t}[X^2_{t,t+h}]\E_{t}[Z^2_{t,t+h}]}{\E_{t}[X_{t,t+h}]} = \frac{\E_{t}[X^2_{t,t+h}]\E_{t}[X_{t,t+h}]}{\E_{t}[X_{t,t+h}]},\]
where the last equality follows from the fact that $X_{t,t+h}$ is the quadratic variation of $Z$. Since $X_{t,t+h}$ is an integral with almost surely finite integrant $V_t$, we have $\E_{t}[X^2_{t,t+h}] \to 0$, as $h \to 0$. In the constrained case, we have that
\[\lim_{h\to0} c = \bs{\omega} \bs{U}_{t} - \lim_{h\to0} \bs{\omega} (\text{diag}({\bs{x}}) \hat{\bs{X}}_{t,t_{h}}(0) + \lim_{h\to 0} g_0(t+h) = \bs{\omega} \bs{U}_{t} +  g_0(t) = V_t,\]
and since $\alpha \to 0$, we have $ \lim_{h\to 0}\frac{\nu \alpha \sum_{n=1}^N \omega_n}{c} = 0$, if $V_t > 0$. The constraint algorithm cannot attain $V_t = 0$ due to the constraint, which proves the assertion. 
\end{proof}

\end{lem}

\section{Numerical Results}
\label{sec:num}
In this section, we run several numerical experiments to assess the properties of the novel method derived in the previous section. We discuss the estimation of the constrained slope parameter $(\beta^C_i)$ in \Cref{sec:slope_estimation}, and then analyze the convergence of the method in experiments in \Cref{sec:convergenceE,sec:convergenceV,sec:Abi} to various benchmarks. \Cref{sec:convergenceE} analyzes the sampling error compared to the Euler method, to show that the C-LP method indeed converges. In \Cref{sec:convergenceV} we show how the method can be used to price volatility derivatives. \Cref{sec:Abi} then highlights the difference in convergences between the C-LP method and an alternative IVI method. 

\subsection{Slope Estimation}
\label{sec:slope_estimation}
The main ingredient of the C-LP method is the linear approximation of the integrated variance, given by the results of \Cref{lem:exp} and \Cref{lem:cov}. We visualize the approximation for a set of parameters by comparing the linear equation to benchmark values created with the EM method with a high number of steps. \Cref{fig:reg} show how values for $\alpha = \E_{s}[X_{s,t}]$ and $\beta= \frac{\E_{s}[X_{s,t}Z_{s,t}]}{\E_{s}[X_{s,t}]}$ provide the linear projection between the two stochastic variables for two sets of parameters.
\begin{figure}[H]
    \centering
    \includegraphics[width=\linewidth]{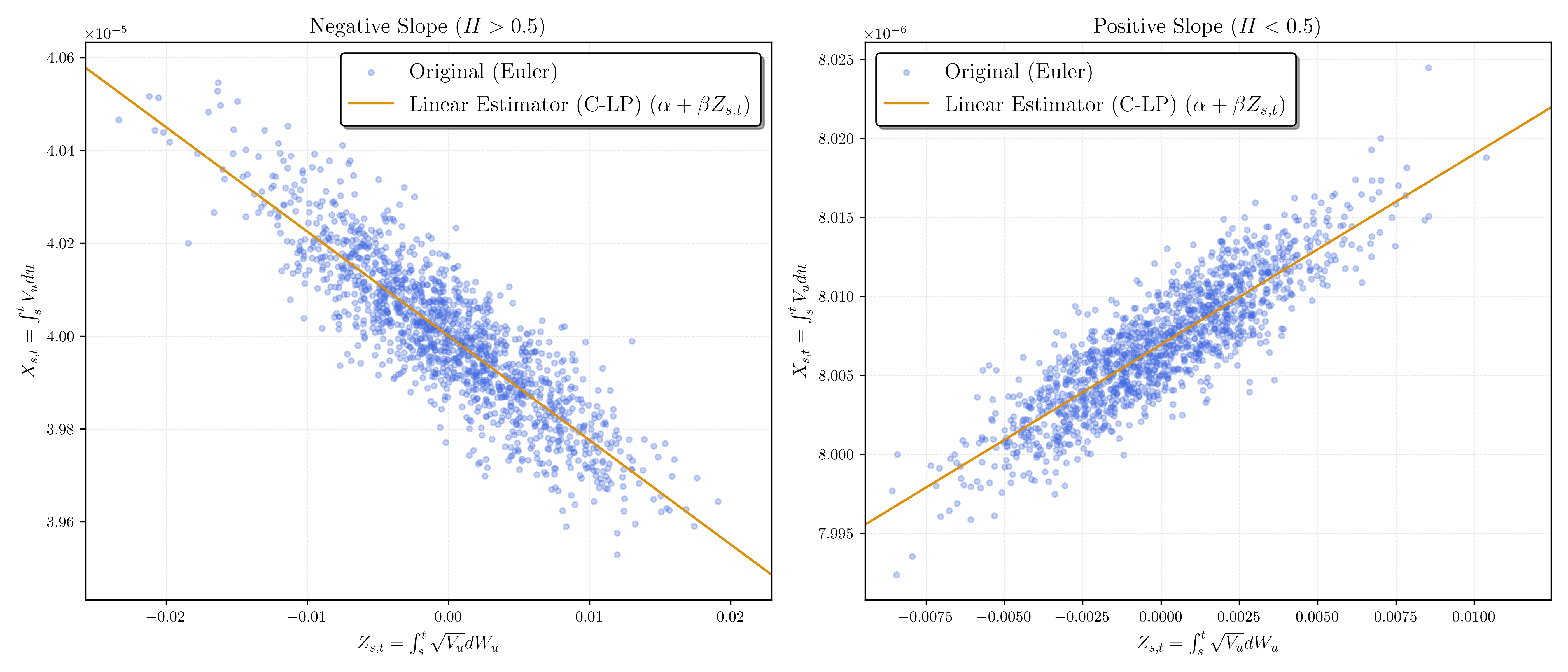}
    \caption{The mean and covariance as computed above offer the best linear approximation between $Z_{s,t}$ and $X_{s,t}$. Note that the start value $V_s$, which we condition on, is deterministic, such that $\alpha,\beta$ are fixed. }
    \label{fig:reg}
\end{figure}
The error of the approximation is the residual term $e_{s,t} := X_{s,t} - (\alpha + \beta Z_{s,t})$. We can derive its variance by calculating the expectation:
\begin{align*}
    \E_{s}[e_{s,t}^2] &= \E_s\left[\left(X_{s,t} - (\alpha + \beta Z_{s,t}\right))^2\right]\\
    &= \E_s\left[X_{s,t}^2 + \alpha^2 + (\beta Z_{s,t})^2 + 2\alpha\beta Z_{s,t} - 2 \alpha X_{s,t} - 2\beta X_{s,t} Z_{s,t}\right]\\
    &= \E_s\left[X_{s,t}^2\right] + \alpha^2 + \beta^2\E_s\left[ X_{s,t}\right] - 2 \alpha \E_s\left[X_{s,t}\right] - 2\beta \E_s\left[X_{s,t} Z_{s,t}\right]\\
    &= \E_s\left[X_{s,t}^2\right] + \alpha^2 + \alpha\beta^2 - 2 \alpha^2 - \alpha 2\beta^2 \\
    &= \E_s\left[X_{s,t}^2\right] - (\alpha\beta^2 + \alpha^2).
\end{align*}
The terms $\alpha\beta^2 + \alpha^2$ are exactly equal to the second moment of the IG random variable $\E_{s}[\hat{X}_{s,t}^2] = \alpha\beta^2 + \alpha^2$. The remaining expectation $\E_s\left[X_{s,t}^2\right]$ is not analytically available. Nevertheless, the calculation shows that the estimation error depends on how well the second moment of the IG random variable approximates the true variance. Since the exact implication of each parameter on this error is unknown, we numerically estimate the impact by parameter variation for a given set of parameters. We consider a fixed set of parameters $\Theta  = (\lambda, \sigma, v_0 , \theta ,\rho , H , N)$ and numerically compute the variation
\begin{equation}
\frac{\partial  \E_{s}[e_{s,t}^2] }{\partial \Theta_k} \approx \frac{E_{\Theta +  \Delta_k} - E_{\Theta}}{ \Delta_k}, \quad k\leq K.
\end{equation}
where $E_{\Theta} = \E_{s}[e_{s,t}^2]$ with parameter set $\Theta$ and $\Theta+ \Delta_k$ represents a variation $\Delta_k$ in the $k$-th parameter. We choose $ \Delta_k = 0.001$ for all parameters except for the discrete parameter $N$, where we choose $ \Delta_N = 1$. \Cref{tbl:variation} shows the computation at $\Theta = (0.25 , 0.1 , 0.02  ,0.5  ,0.7 ,  0.3,  5)$, both in absolute and relative sensitivity. The absolute level of $V_t$ has a positive impact on $\E_{s}[e_{s,t}^2]$ through the positive coefficient of $v_0$ and $\theta$. Furthermore, the vol-of-vol $(\nu)$ parameter has a positive coefficient as well.
\begin{table}[H]
    \centering
    \begin{tabular}{c r r r r r r}
    \hline
         & $\lambda$ & $\nu$ & $v_0$ & $\theta$ & $h$ & $N$ \\ \hline
        $\frac{E_{\Theta_0+  \Delta_k} - E_{\Theta_0}}{ \Delta^k}$ 
            & 0.02891 & 0.01100 & 0.0164 & 0.0402 & 0.0098 & -0.0003\\
        $\frac{E_{\Theta+  \Delta_k} - E_{\Theta}}{ \Delta_k} \cdot \Theta_k$ 
            & 11.52\% & 10.89\% & 78.07\% & 8.03\% & 3.25\% & -0.01\% \\ \hline
    \end{tabular}
  
    \caption{Sensitivities of $\E_{s}[e_{s,t}^2]$ w.r.t. parameters, excl. $\rho$ as it has no impact on $X_{s,t}$.}
      \label{tbl:variation}
\end{table}
The slope of the linear approximation is possibly constrained due to the positivity condition of the variance. Here, we visualize the impact of the constrained slope. \Cref{fig:slope_infeas} shows an example of an unconstrained linear projection, which can lead to samples $\hat{V}_{t_{i+1}}$ that are negative, and hence infeasible.
\begin{figure}[H]
    \centering
    \includegraphics[width=\linewidth]{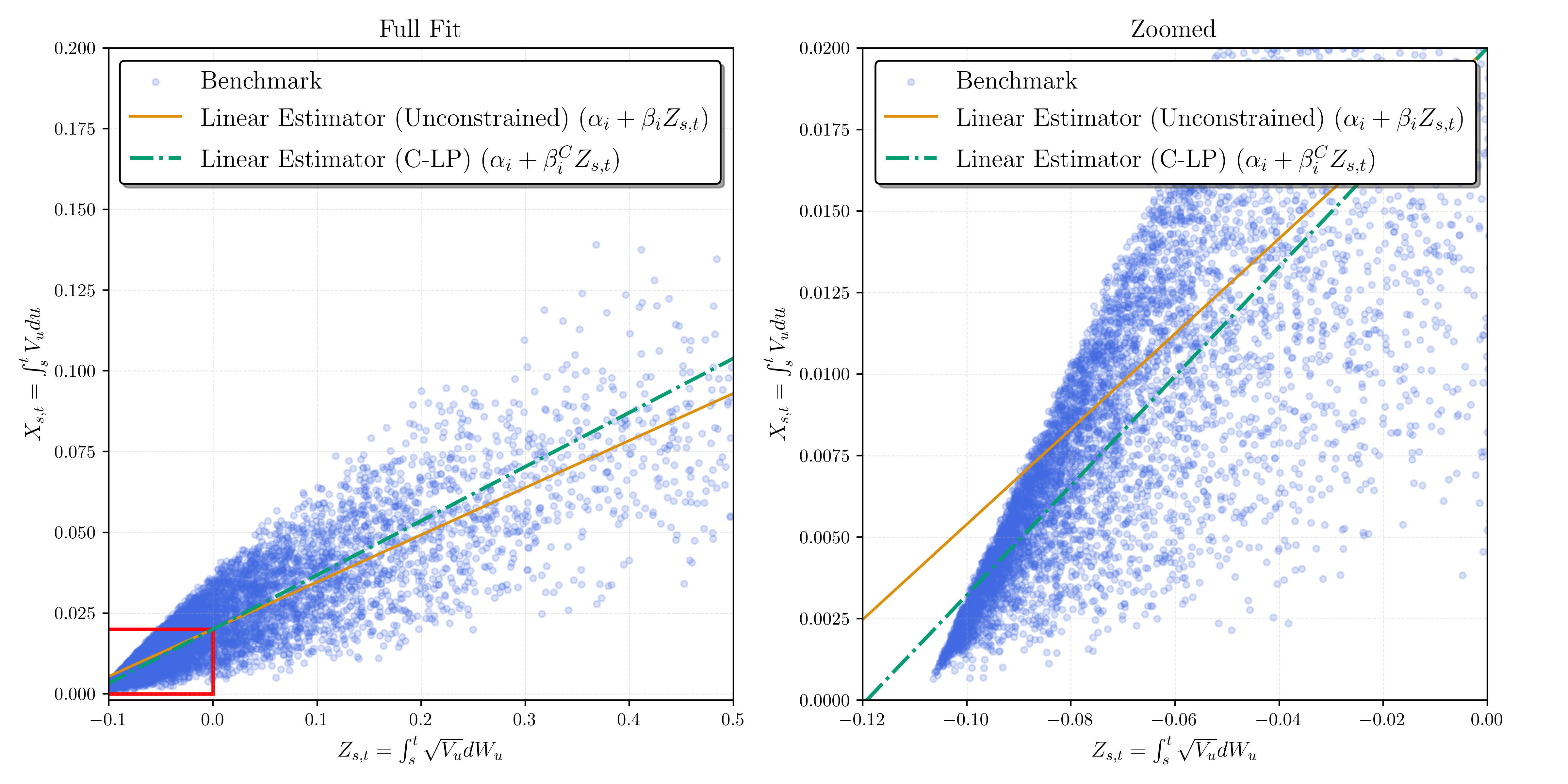}
    \includegraphics[width=0.7\linewidth]{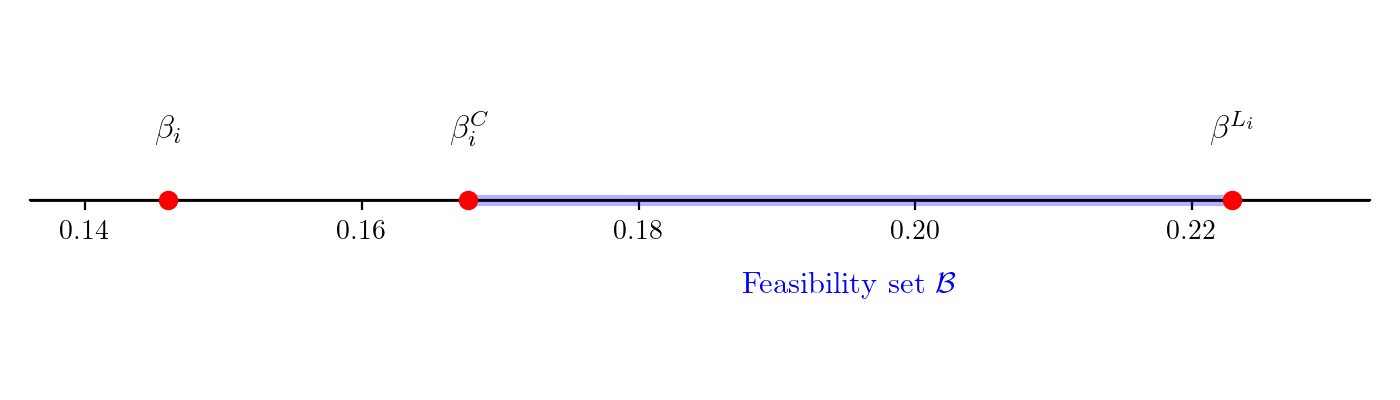}
    \caption{For the extreme parameter set $\nu=0.3, v_0 = 0.02, \theta = 0.02$, the unconstrained beta can cause sample combinations of $\hat{X}_{t_i,t_{i+1}}, \hat{Z}_{t_i,t_{i+1}}$ leading to negative variances. The unconstrained beta is therefore not in the feasibility set, which lies between $\beta_i^C$ and the upper bound $\beta^{L_i}$. }
    \label{fig:slope_infeas}
\end{figure}

\subsection{Convergence of Integrated Variance }
\label{sec:convergenceE}
To highlight the effectiveness of the novel algorithm, we analyze the speed of convergence of the algorithm. We run the algorithm for a varying number of time steps, starting from very large time steps, and compare the statistical properties of the resulting samples to a benchmark set. The benchmark set we obtain from an Euler discretization with very small time steps ($1\,000$ steps), which we consider sufficiently small to be close to the actual distribution. We assess the convergence of the statistical properties of the final quantity of interest $X_T := \int_{t_0}^T V_u \du$ at a final time $T > t_0$. We consider the empirical quantities for $\E[X_T ]$ and $\var(X_T)$ for a varying number of time steps for a final time $T=5$ between the EM scheme and the C-LP scheme.  We run the experiment on a diverse set of parameters, such as varying numbers of $N$ and the Hurst parameter $H$. Set 3 is an example of a realistic parameter set according to~\cite{abi2019lifting}. In \Cref{tab:convergence-params}, we display the selected parameter sets for the experiment
\begin{table}[H]
\centering
\begin{tabular}{cccccccc}
\toprule
Name&$\lambda$ & $\nu$ & $v_0$ & $\theta$ & $\rho$ & $H$ & $N$ \\
\midrule
Set 1&0.25 & 0.1 & 0.02 & 0.5 & 0.7  & 0.3 & 5  \\
Set 2&0.1 & 0.2 & 0.1 & 0.7 & -0.7 & 0.1 & 10 \\
Set 3 \cite{abi2019lifting} &0 & 0.31 & 0.1  & 0.02 & 0.7  & 0.3 & 20 \\
\bottomrule
\end{tabular}
\caption{Lifted Heston parameter sets used in numerical experiments.}
\label{tab:convergence-params}
\end{table}
In \Cref{fig:x_prop_clp_euler}, we plot the absolute differences of the final quantities to the benchmark. The C-LP scheme is able to reach a small absolute error, close to the standard error, with large step sizes. Depending on the parameter set, the EM scheme only converges when step sizes are below $\Delta t = 0.2$, or even smaller. For the last set, a $\Delta t = 0.025$ is not sufficient to provide meaningful results. On the other hand, the C-LP provides excellent results for step sizes of $\Delta t = 2.15$. 
\begin{figure}[H]
    \centering
    \includegraphics[width=1\linewidth]{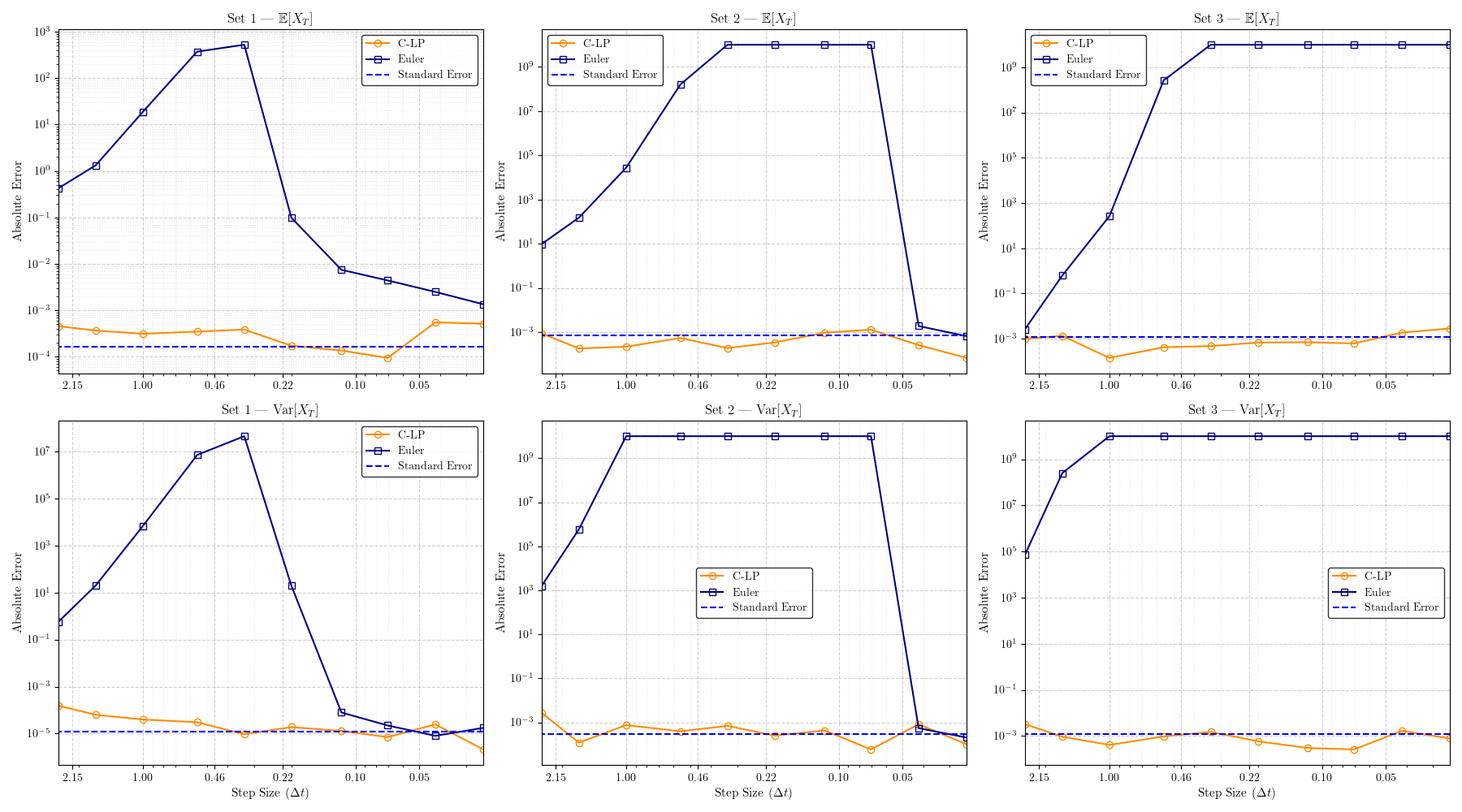}
    \caption{EM scheme vs C-LP: $\E[X_T]$ and $\var(X_T)$, $T = 5$. The errors are capped at $10^3$.}
    \label{fig:x_prop_clp_euler}
\end{figure}

\subsection{Pricing Volatility Options}
\label{sec:convergenceV}
In the second experiment, we test a practical application of volatility models with a rough component. The pricing volatility of derivatives, such as options on the VIX index, depends directly on the distributive properties of the integrated variance process. We analyze the pricing accuracy of the novel simulation scheme to price such derivatives. We define the VIX index as a random process $VIX_t, t \geq t_0$, defined as the quantity
\begin{equation}
    \label{eq:vix}
    VIX_T := \sqrt{\frac{1}{\Theta}\int_{T}^{T+ \Theta} \E[V_u|\F_T] \du } = \sqrt{\frac{1}{\Theta}{\E[X_{T, T +\Theta }|\F_T]}},
\end{equation}
where $\Theta$ is a predetermined time horizon of 1 month. Note that $VIX_t$ is not $\F_t$ but $\F_{t +\Theta}$-adapted.

We conduct a numerical experiment to show the speed of convergence of the C-LP scheme in pricing $VIX$ derivatives. Using the parameter sets from \Cref{tab:convergence-params}, we use the novel numerical scheme to price European VIX options with time-to-maturity $T=1$ and show how the implied volatility of the model price compares to a model price obtained through an Euler simulation with a comparable number of steps. The prices are estimated from a Monte-Carlo simulation, where we simulate the variance process $V_t$ until time $T+ \Theta$. To ensure that we can easily estimate the $VIX_T$, we choose the number of time steps such that $T=1$ is on the time grid, i.e., such that $T \mod \Delta t = 0$. Since $\Theta$ is one month, we need to simulate the horizon $t_0=0$ to $T+\Theta = 1 + \frac{1}{12}$. For that reason, we choose time steps as multiples of $13$. For the first two parameter sets, we show the time steps $13,26$ and $39$. For the third set, we include $72$, since the Euler scheme does not produce usable results for small time steps due to the numerical instability.
\begin{table}[H]
\centering
        \begin{tabular}{c|c}
Number of Steps & $\Delta t$   \\
\hline
13              & 1/12    \\
26              & 1/24      \\
39              & 1/36  \\
78              & 1/74  
\end{tabular}

\caption{Equidistant grid: Number of steps and step sizes for $T=1 + \frac{1}{12}$}
\label{tbl:dt}
\end{table}
\Cref{tbl:dt} shows an overview of number of steps and $\Delta t$ for the experiment.
\begin{figure}[H]
    \centering
    \includegraphics[width=\linewidth]{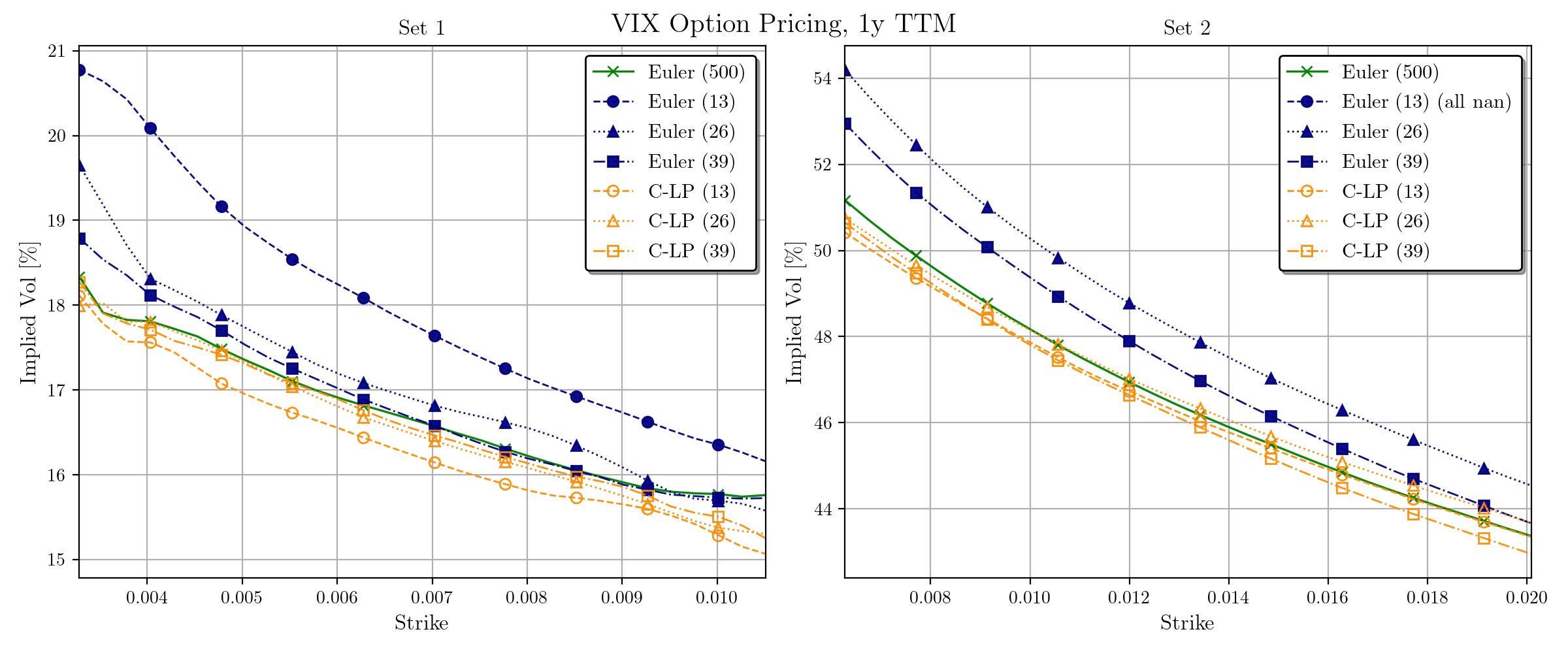}
        \includegraphics[width=0.5\linewidth]{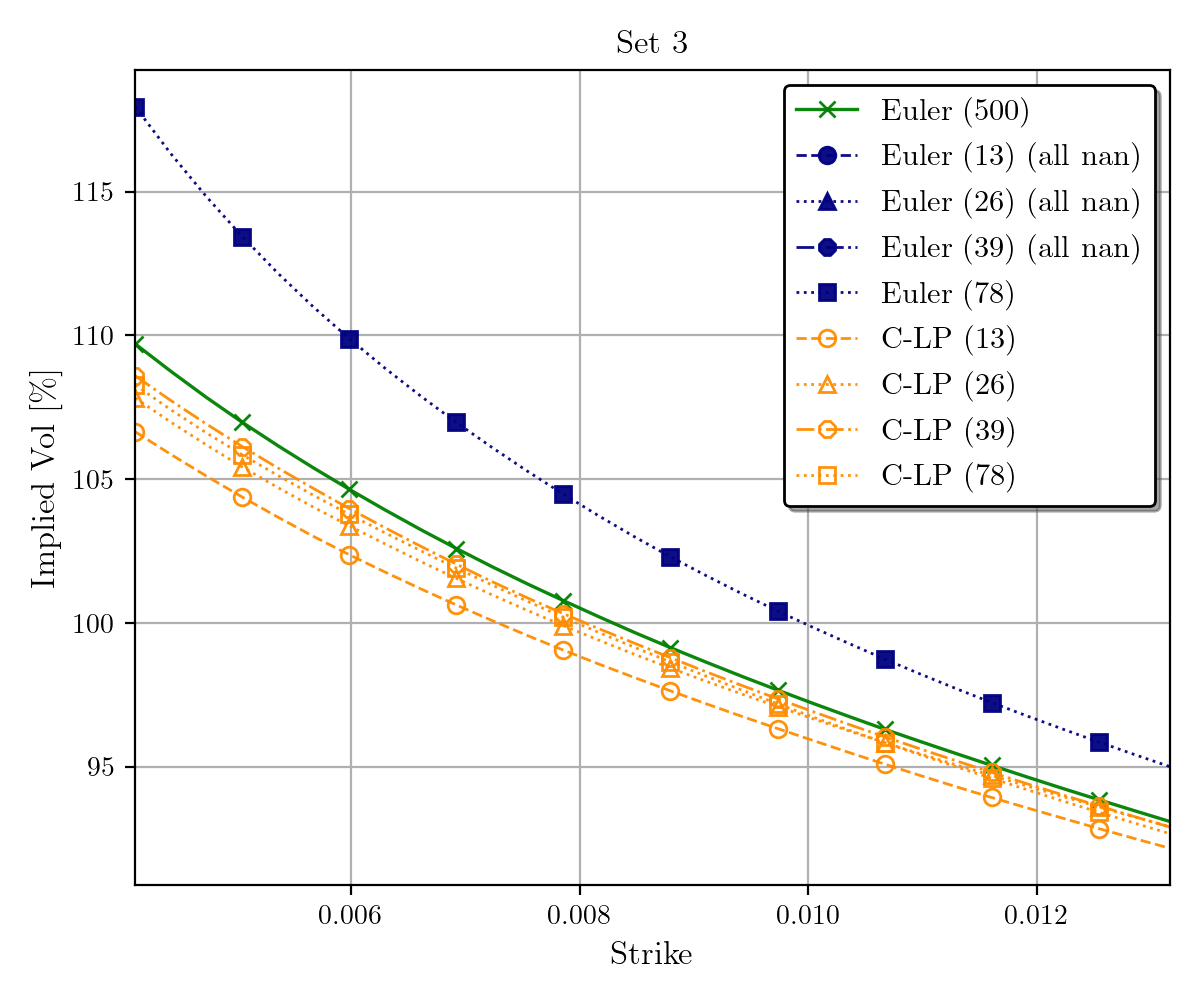}

    \caption{Pricing of European VIX options - C-LP vs Euler Scheme}
    \label{fig:vix_options}
\end{figure}
\Cref{fig:vix_options} shows the results of the experiment. As expected from the convergence of the moments of the integrated variance, when pricing the $VIX$, the C-LP scheme produces reasonable implied volatilities for much smaller time steps compared to the Euler scheme. Furthermore, due to the constrained formulation, the scheme is valid for any number of time steps.
\subsection{Alternative Integrated Variance Schemes}
\label{sec:Abi}
Abi Jaber and Attal~\cite{jaber2025simulating} derived a general method for simulating rough Heston-type models, which can be applied to the lifted Heston model using a specific kernel choice to obtain a similar scheme. In this section, we briefly compare the numerical properties of the two schemes and highlight their differences. While the general framework of Abi and Attal offers flexibility in modeling (possible for other kernels than the rough Heston model), and also significantly improves convergence for larger steps compared to the Euler scheme, the formulation of the C-LP as the \emph{best} possible approximation suggests a faster convergence and a possibility for even larger time steps. 

While both methods rely on sampling of the integrated variance through an inverse Gaussian distribution, the main difference lies in how to obtain the linear approximation. In the Abi/Attal scheme, the linear relation is derived from the dynamics of the process. The authors use a right-point approximation of certain integrals that integrate over the entire time step. This implies that the approximation error is proportional to the time step size. In the C-LP scheme, we avoid such an approximation and obtain the relation as a least-square optimization, meaning that the approximation error only depends on the inherent non-linearity between the two random variables. \Cref{fig:slope_abi_clp} displays the impact of the right-point approximation in examples of computed linear estimators versus a true distribution obtained from the EM scheme.
\begin{figure}[H]
    \centering
    \includegraphics[width=\linewidth]{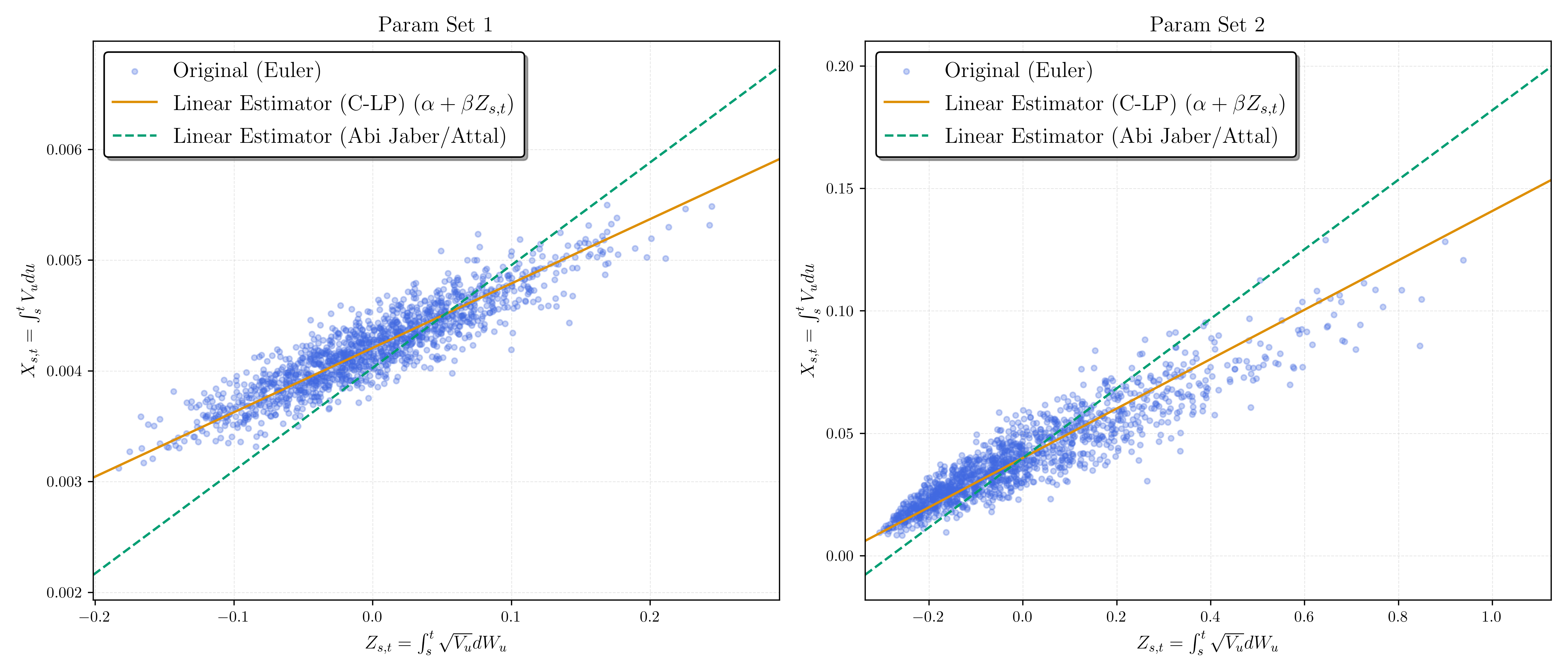}
    \caption{Linear Estimations between C-LP and Abi/Attal method.}
    \label{fig:slope_abi_clp}
\end{figure}

To highlight the impact of choosing a suboptimal linear equation, we conduct a similar study to examine the convergence of the integrated variance. \Cref{fig:x_prop_clp_abi} shows the results of the simulation for the C-LP scheme, where we plot the absolute error of the simulation to the benchmark (small step EM simulation), and display the standard error of the expectation and variance given the $200\,000$ samples. We notice that the C-LP scheme provides small errors for extremely large time steps, such as $\Delta t = 5$. Since the method is moment matching, the expectation matches the benchmark up to a standard error. The variance of $X_{T}$ converges fast to the benchmark, with significantly reduced error compared to the Abi/Attal method. Furthermore, due to the constrained formulation of the C-LP, the method yields robust results for set 3, which other methods fail to achieve due to paths with negative variances.
\begin{figure}[H]
    \centering
    \includegraphics[width=1\linewidth]{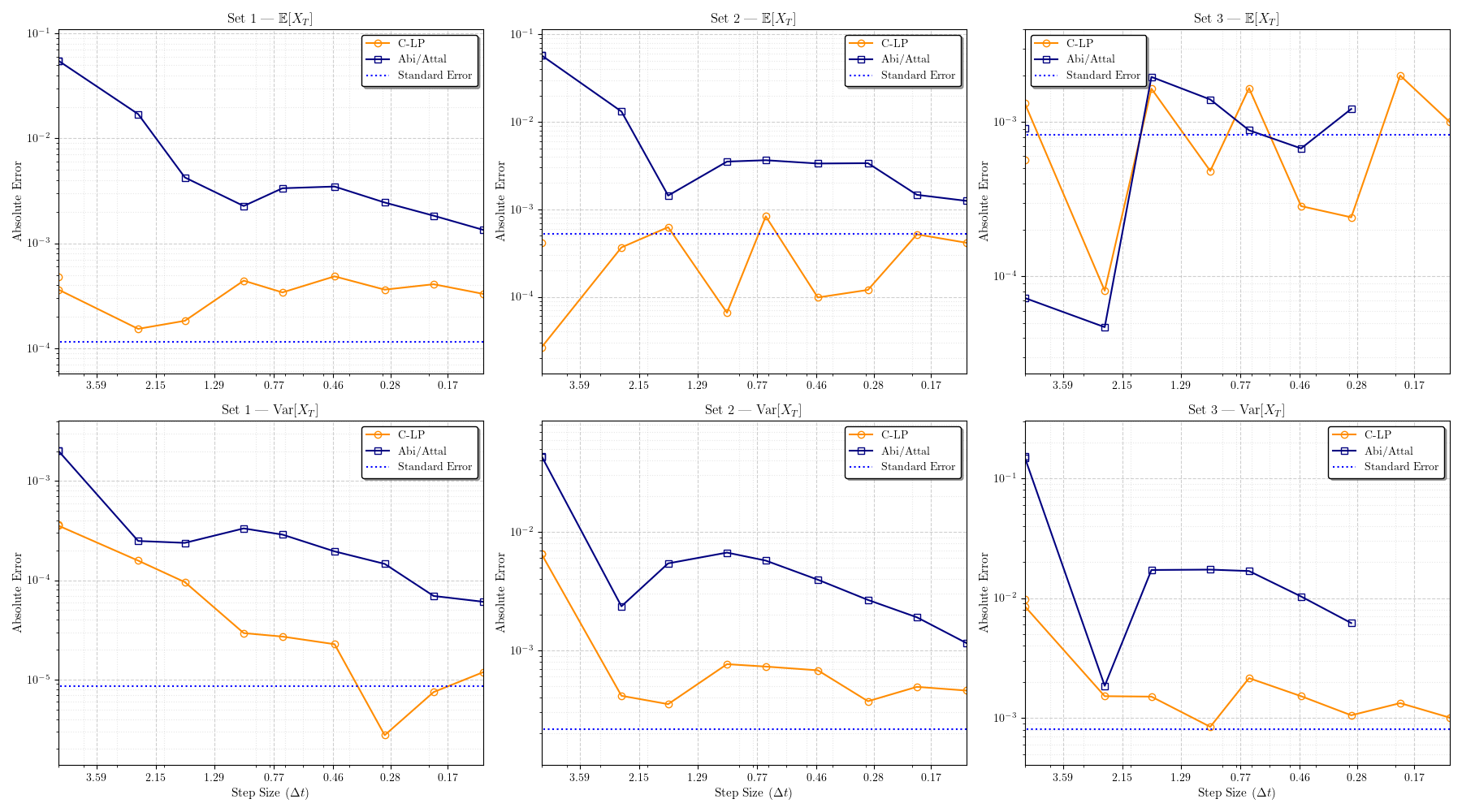}
    \caption{Abi/Attal vs C-LP: $\E[X_T]$ and $\var(X_T)$, $T = 5$.}
    \label{fig:x_prop_clp_abi}
\end{figure}

\section{Conclusions}
The lifted Heston model is a stochastic volatility model, where the volatility is driven by a set of state processes with a common driver. Although Markovian, the multidimensional nature of the process is challenging, as traditional methods require a small step approximation of the multidimensional process. Incorporating the insights from the well-studied simulation of the regular Heston model, we develop a novel scheme capable of taking larger time steps. The scheme establishes a linear relation between the integrated variance process $\int_s^t V_u \du$ and the stochastic driver $\int_s^t \sqrt{V_u}\dW_u$ using an $L^2$ projection. To ensure that the scheme remains well-defined, this projection is then constrained to a feasible parameter space, and is hence called the \emph{constrained linear projection method} (C-LP). This method provides the best possible scheme in the class of integrated variance implicit schemes, as it optimizes the linear approximation, which is the only source of approximation.

In the numerical experiments in \Cref{sec:num}, we then show that the scheme is capable of simulating the lifted Heston model for much larger steps than the Euler scheme, providing robust simulations for any valid parameter set. Additionally, we demonstrate that this method outperforms similar schemes, such as the Abi/Attal method, as it employs a suboptimal linear approximation of the integrated variance and stochastic driver. 
\bibliography{bib.bib}

\begin{thebibliography}{10}

\bibitem{abi2019lifting}
{\sc E.~Abi~Jaber}, {\em Lifting the {H}eston model}, Quantitative finance, 19 (2019), pp.~1995--2013.

\bibitem{jaber2024simulation}
\leavevmode\vrule height 2pt depth -1.6pt width 23pt, {\em Simulation of square-root processes made simple: applications to the {H}eston model}, arXiv preprint arXiv:2412.11264,  (2024).

\bibitem{jaber2025simulating}
{\sc E.~Abi~Jaber and E.~Attal}, {\em Simulating integrated {V}olterra square-root processes and {V}olterra {H}eston models via inverse gaussian}, arXiv preprint arXiv:2504.19885,  (2025).

\bibitem{alfonsi2005discretization}
{\sc A.~Alfonsi}, {\em On the discretization schemes for the {CIR} (and {B}essel squared) processes}, Monte Carlo Methods Appl., 11 (2005), pp.~355--384.

\bibitem{andersen2008simple}
{\sc L.~Andersen}, {\em Simple and efficient simulation of the {H}eston stochastic volatility model}, Journal of Computational Finance, 11 (2008), pp.~1--43.

\bibitem{bayer2023markovian}
{\sc C.~Bayer and S.~Breneis}, {\em Markovian approximations of stochastic {V}olterra equations with the fractional kernel}, Quantitative Finance, 23 (2023), pp.~53--70.

\bibitem{bayer2016pricing}
{\sc C.~Bayer, P.~Friz, and J.~Gatheral}, {\em Pricing under rough volatility}, Quantitative Finance, 16 (2016), pp.~887--904.

\bibitem{bayer2023rough}
{\sc C.~Bayer, P.~K. Friz, M.~Fukasawa, J.~Gatheral, A.~Jacquier, and M.~Rosenbaum}, {\em Rough volatility}, SIAM, 2023.

\bibitem{bergomi2015stochastic}
{\sc L.~Bergomi}, {\em Stochastic volatility modeling}, CRC press, 2015.

\bibitem{broadie2006exact}
{\sc M.~Broadie and {\"O}.~Kaya}, {\em Exact simulation of stochastic volatility and other affine jump diffusion processes}, Operations research, 54 (2006), pp.~217--231.

\bibitem{deelstra1998convergence}
{\sc G.~Deelstra and F.~Delbaen}, {\em Convergence of discretized stochastic (interest rate) processes with stochastic drift term}, Applied stochastic models and data analysis, 14 (1998), pp.~77--84.

\bibitem{dubins1965continuous}
{\sc L.~E. Dubins and G.~Schwarz}, {\em On continuous martingales}, Proceedings of the National Academy of Sciences, 53 (1965), pp.~913--916.

\bibitem{el2019characteristic}
{\sc O.~El~Euch and M.~Rosenbaum}, {\em The characteristic function of rough {H}eston models}, Mathematical Finance, 29 (2019), pp.~3--38.

\bibitem{folks1978inverse}
{\sc J.~L. Folks and R.~S. Chhikara}, {\em The inverse {G}aussian distribution and its statistical application - a review}, Journal of the Royal Statistical Society Series B: Statistical Methodology, 40 (1978), pp.~263--275.

\bibitem{grzelak2019stochastic}
{\sc L.~A. Grzelak, J.~A. Witteveen, M.~Suarez-Taboada, and C.~W. Oosterlee}, {\em The stochastic collocation {M}onte {C}arlo sampler: highly efficient sampling from `expensive' distributions}, Quantitative Finance, 19 (2019), pp.~339--356.

\bibitem{heston1993closed}
{\sc S.~L. Heston}, {\em A closed-form solution for options with stochastic volatility with applications to bond and currency options}, The review of financial studies, 6 (1993), pp.~327--343.

\bibitem{kloeden2013numerical}
{\sc P.~Kloeden and E.~Platen}, {\em Numerical Solution of Stochastic Differential Equations}, Stochastic Modelling and Applied Probability, Springer Berlin Heidelberg, 2013.

\bibitem{lord2010comparison}
{\sc R.~Lord, R.~Koekkoek, and D.~V. Dijk}, {\em A comparison of biased simulation schemes for stochastic volatility models}, Quantitative Finance, 10 (2010), pp.~177--194.

\bibitem{ma2022fast}
{\sc J.~Ma and H.~Wu}, {\em A fast algorithm for simulation of rough volatility models}, Quantitative Finance, 22 (2022), pp.~447--462.

\bibitem{Oosterlee_Grzelak_2020}
{\sc C.~W. Oosterlee and L.~A. Grzelak}, {\em Mathematical Modeling and Computation in Finance: With exercises and Python and MATLAB computer codes}, World Scientific Publishing Europe Ltd, 2019.

\bibitem{richard2023discrete}
{\sc A.~Richard, X.~Tan, and F.~Yang}, {\em On the discrete-time simulation of the rough {H}eston model}, SIAM Journal on Financial Mathematics, 14 (2023), pp.~223--249.

\end{thebibliography}
\appendix
\section{Proof of \Cref{eq:expecation}}
\label{app:proof_mean}
Unlike the regular Heston model, the unconditional expectation of the variance process is not straightforward in the lifted Heston Model. While the value depends on the exact implementation of $g_0(t)$, we can derive an integral equation to obtain a general form of the solution. To obtain it, we first remove the $U_t^n$ term in (\ref{eq:lifted_heston_end}) using It\^o's lemma. We obtain for any $n \leq N$
\begin{align*}
    \d \e^{x_n t}U_t^n &= x_n\e^{x_n t} U_t^n + \e^{x_n t} \d U_t^n \\
    &=  -x_n\e^{x_n t} U_t^n \dt + x_n\e^{x_n t} U_t^n \dt  - \e^{x_n t} \lambda V_t \dt + \e^{x_n t} \nu \sqrt{V_t}\dW_t^2\\
    &= - \e^{x_n t}\lambda V_t \dt + \e^{x_n t}\nu \sqrt{V_t}\dW_t^2,
\end{align*}
and therefore for any $T \geq t_0$,
\begin{equation*}
    U_T^n = \int_{t_0}^T - \lambda \e^{-x_n (T-t)}V_t \dt + \int_{t_0}^T \e^{-x_n (T-t)} \sqrt{V_t}\dW_t^2.
\end{equation*}
Taking the expectation in \Cref{eq:lifted_heston_middle}, we obtain the Volterra equation
\begin{align*}
    \E_{t_i}[V_T] &= g_0(T) + \bs{\omega} \E_{t_i}[\bs U_T]\\
    &=  g_0(T) - \lambda\sum_{n=1}^N \omega_n \int_{t_0}^T \e^{-x_n (T-t)} \E[V_t] \dt .
\end{align*}
Plugging in the choice of $g_0$ as in \Cref{eq:fwd_curve_simple} we obtain
\begin{align*}
    \E[V_T] &=  V_0 + \lambda \theta \sum_{n=1}^N \frac{\omega_n}{x_n} \left[1 - \e^{-x_n (T-t_0)}\right]  - \lambda\sum_{n=1}^N \omega_n \int_{t_0}^T \e^{-x_n (T-t)} \E[V_t] \dt \\
    &=V_0 + \sum_{n=1}^N \omega_n \int_{t_0}^T \lambda(\theta - \E[V_t]) \e^{-x_n (T-t)}\dt,
\end{align*}
which shows the importance of $\theta$ in controlling the long-term mean of $V_t$.
\end{document}